\documentclass{article}

\usepackage[preprint]{neurips_2025}

\usepackage[utf8]{inputenc} 
\usepackage[T1]{fontenc}    
\usepackage{hyperref}       
\usepackage{url}            
\usepackage{booktabs}       
\usepackage{amsfonts}       
\usepackage{nicefrac}       
\usepackage{microtype}      
\usepackage{xcolor}         

\usepackage{amsmath}
\usepackage{amsthm}
\usepackage{enumitem}
\usepackage{wrapfig}
\usepackage{relsize}
\usepackage{graphicx}   
\usepackage{subcaption}

\newtheorem{theorem}{Theorem}[section]
\newtheorem{proposition}[theorem]{Proposition}
\newtheorem{lemma}[theorem]{Lemma}

\newtheorem{definition}[theorem]{Definition}
\newtheorem{assumption}[theorem]{Assumption}
\newtheorem{remark}[theorem]{Remark}

\title{Shapley Machine: A Game-Theoretic Framework for N-Agent Ad Hoc Teamwork}

\author{%
  Jianhong Wang\thanks{Jianhong Wang is a visiting researcher at the Centre for AI Fundamentals at University of Manchester. $^\dagger$ Correspondence to Jianhong Wang and Yang Li.}\ \ $^{\dagger}$ \\
  INFORMED-AI Hub \\
  University of Bristol, UK \\
  \texttt{jianhong.wang@bristol.ac.uk} \\
  \And
  Yang Li$^\dagger$ \\
  Centre for AI Fundamentals \\
  University of Manchester, UK \\
  \texttt{yang.li-4@manchester.ac.uk} \\
  \And
  Samuel Kaski \\
  Centre for AI Fundamentals \\
  University of Manchester, UK \\
  \texttt{samuel.kaski@manchester.ac.uk} \\
  \And
  Jonathan Lawry \\
  School of Engineering Mathematics and Technology \\
  University of Bristol, UK \\
  \texttt{J.Lawry@bristol.ac.uk}
}

\begin{document}

\maketitle

\begin{abstract}
  Open multi-agent systems are increasingly important in modeling real-world applications, such as smart grids, swarm robotics, etc. In this paper, we aim to investigate a recently proposed problem for open multi-agent systems, referred to as n-agent ad hoc teamwork (NAHT), where only a number of agents are controlled. Existing methods tend to be based on heuristic design and consequently lack theoretical rigor and ambiguous credit assignment among agents. To address these limitations, we model and solve NAHT through the lens of cooperative game theory. More specifically, we first model an open multi-agent system, characterized by its value, as an instance situated in a space of cooperative games, generated by a set of basis games. We then extend this space, along with the state space, to accommodate dynamic scenarios, thereby characterizing NAHT. Exploiting the justifiable assumption that basis game values correspond to a sequence of n-step returns with different horizons, we represent the state values for NAHT in a form similar to $\lambda$-returns. Furthermore, we derive Shapley values to allocate state values to the controlled agents, as credits for their contributions to the ad hoc team. Different from the conventional approach to shaping Shapley values in an explicit form, we shape Shapley values by fulfilling the three axioms uniquely describing them, well defined on the extended game space describing NAHT. To estimate Shapley values in dynamic scenarios, we propose a TD($\lambda$)-like algorithm. The resulting reinforcement learning (RL) algorithm is referred to as Shapley Machine. To our best knowledge, this is the first time that the concepts from cooperative game theory are directly related to RL concepts. In experiments, we demonstrate the effectiveness of Shapley Machine and verify reasonableness of our theory.
\end{abstract}

\section{Introduction}
\label{sec:introduction}
    Multi-agent systems (MAS) have emerged as a prominent research area, effectively modeling and addressing numerous practical applications involving collaborative tasks, such as distributed electric power network control~\cite{wang2021multi}, railway network distribution management~\cite{zhang2024improving}, and swarm robotics~\cite{nayak2023scalable,li2024holadrone}. Reinforcement learning (RL)~\cite{sutton1998reinforcement}, particularly multi-agent reinforcement learning (MARL), has shown great promise in solving MAS problems. However, MARL typically struggles with real-world conditions such as openness, where the number of uncontrolled agents can vary, and generalization, involving unknown teammate behaviors. To address these limitations, a new problem termed n-agent ad hoc teamwork (NAHT) was recently introduced by \cite{wang2025n}, extending MARL to scenarios with varying and potentially unknown teammates. In this paper, we aim to establish a theoretical basis for modeling open multi-agent systems and propose a novel algorithm for NAHT .
    
    Given the novelty of this problem, only an initial practical solution called POAM~\cite{wang2025n} has been proposed. POAM enhances the IPPO algorithm~\cite{de2020independent} by incorporating an embedding vector, which serves as a continuous representation derived from each agent's historical observations using an encoder-decoder model. This embedding aims to predict the potential behaviors of other agents from each agent’s unique perspective. Despite POAM demonstrating superior performance compared to existing MARL and ad hoc teamwork algorithms, it has several notable limitations: (1) The algorithm is designed heuristically without rigorous theoretical foundations, undermining the trustworthiness crucial for ensuring reliability and safety in multi-agent systems~\cite{hammond2025multi}. (2) While controlled agents within a team receive different evaluation values during training to guide their decentralized policies, it remains unclear whether these values accurately reflect individual agents' contributions to team performance. As with other challenges in MARL, this uncertainty can adversely affect learning dynamics and efficiency~\cite{panait2005cooperative}, which is central to online learning for safety-critical systems. (3) The value functions (critics) employed in POAM are trained using TD($\lambda$)~\cite{sutton1988learning}. Beyond the RL training benefits, the broader relevance of TD($\lambda$) within open multi-agent systems remains an open question.
    
    Shapley value is a profound payoff distribution scheme in cooperative game theory, which has been broadly applied in MARL~\citep{wang2020shapley,li2021shapley,han2022multiagent,wang2022shaq,pmlr-v202-li23au}. Almost all the previous work concentrated on using the explicit form of Shapley values to shape rewards or value functions. On the other hand, Shapley values have been merely used to address cooperative MARL with deterministic number of agents, due to the lack of a corresponding theory on openness. To bridge this gap, this paper will establish a theory underlying Shapley values for open multi-agent systems, and thus derive an algorithm that solves NAHT, depending on the axiomatic characterization of Shapley values. More importantly, this algorithms is shown to have close relationship to the TD($\lambda$) as established in RL literature.

    The main contributions of this paper are as follows: (1) A theoretical model, through the lens of cooperative game theory, is established to describe open multi-agent systems (e.g. the NAHT processes) as instances situated in a vector space of cooperative games.
    (2) We impose a justifiable assumption that maps the basis of the vector space to a sequence of n-step returns with variant horizons. Based on the above theoretical grounding, Shapley value's axiomatic characterization (i.e. Efficiency, Symmetry and Additivity), rigorously defined in the vector space of cooperative games~\cite{dubey1975uniqueness}, is imposed on the value function represented as the form of $\lambda$-return~\cite{sutton1998reinforcement}[Chap. 12], and thus a TD($\lambda$)-like algorithm, referred to as Shapley Machine, is derived. To our knowledge, it is the first time that concepts from cooperative game theory (e.g. Shapley value) are closely related to a reinforcement learning technique (e.g. TD($\lambda$)). (3) The existing algorithm to address NAHT, referred to as POAM, can also be categorized according to the axioms above. 
    Therefore, it sheds light on the potential of designing algorithms for multi-agent systems, following a well-defined theoretical model and axiomatic characterization. This complies with the spirit of algorithm design from the first principle, popularized in the community of geometric deep learning~\cite{bronstein2021geometric}.

    In experiments, we have evaluated the Shapley Machine via the testbeds of NAHT on MPE and SMAC, provided by \cite{wang2025n}. We empirically demonstrate that Shapley Machine can generally outperform the only existing algorithm designed for NAHT, POAM. Additionally, we empirically verify the theory, through a spectrum of case studies. The complete mathematical proofs are left to \textcolor{black}{Appendix~\ref{sec:complete-math-proofs}}. The code is published on \url{https://github.com/hsvgbkhgbv/shapley-machine-naht}.
        
\section{Background}
\label{sec:background}
    \subsection{N-Agent Ad Hoc Teamwork}
    \label{subsec:n-agent-aht}
        In this section, we describe a problem setting for open multi-agent systems, referred to as n-agent ad hoc teamwork (NAHT)~\cite{wang2025n}. The main challenges of NAHT are as follows: (1) coordination with potentially unknown types of teammates (generalization), and (2) coping with a varying number of uncontrolled teammates (openness). A decentralized partially observable Markov decision process (Dec-POMDP)~\cite{oliehoek2016concise} is considered to formalize the problem. There is a set of agents $\tilde{N}$, a state space $\mathcal{S}$, a joint action space $\mathcal{A}=\times_{i \in \tilde{N}} \mathcal{A}_{i}$, a per-agent observation space $\mathcal{O}_{i}$, a transition function $P_{T}: \mathcal{S} \times \mathcal{A} \rightarrow \Delta(\mathcal{S})$, a common reward function $R: \mathcal{S} \times \mathcal{A} \rightarrow \mathbb{R}$, a discount factor $\gamma \in [0, 1]$ and an episode length $T \in \mathbb{Z}^{+}$ indicating the maximum length of an interaction. Each agent receives observations via an observation function $O_{i}: \mathcal{S} \times \mathcal{A} \rightarrow \Delta(\mathcal{O}_{i})$. $\mathcal{H}_{i}$ is defined as an agent $i$'s space of observation and action histories. An agent is equipped with a policy $\pi_{i}: \mathcal{H}_{i} \rightarrow \Delta(\mathcal{A}_{i})$.

        In this problem, there are two groups of agents: a set of controlled agents denoted as $\tilde{C}$ and a set of uncontrolled agents denoted as $\tilde{U}$. As NAHT is an open system, at the beginning of each interactive process (episode), a team of agents $\tilde{N}$, consisting of $\tilde{c}$ agents drawn from the controlled agent pool $\tilde{C}$ and $\tilde{u}$ agents drawn from the uncontrolled agent pool $\tilde{U}$. For simplicity, each controlled agent is characterized by its policy, and the resulting set of controlled agents is expressed as $\tilde{C}(\theta) = \{\pi_{i}^{\theta}\}_{i=1}^{\tilde{c}}$. Therefore, the above sampling procedure is denoted as $X(\tilde{U}, \tilde{C}(\theta))$. The aim of NAHT is learning parameters $\theta$ to solve the following optimization problem:
        \begin{equation}
            \max_{\theta} \mathbb{E}_{\pi^{(\tilde{N})} \sim X(\tilde{U}, \tilde{C}(\theta))} \left[ \sum_{t=0}^{T} \gamma^{t} R_{t} \right],
        \label{eq:naht-objective}
        \end{equation}
        where $\pi^{(\tilde{N})}$ denotes the joint policy of a team of controlled agents. 

    \subsection{A Set of Cooperative Games and Shapley Value}
    \label{subsec:shapley-value}
        \begin{definition}
        \label{def:game-set}
             We first define such a function that $v_{C, s}^{z} = z$, for any $z \in \mathbb{R}$. If $C \subseteq D$, $v_{C, s}^{z}(D) = z$, otherwise, $v_{C, s}^{z}(D) = 0$. For a cooperative game described by a characteristic function $v: 2^{N} \rightarrow \mathbb{R}_{\geq 0}$, we can describe a set of such games over a group of agents $N$ by a set of basis games $\{v_{C}^{1} \ | \ \emptyset \neq C \subseteq N \}$, such as $G = \{ w \ | \ w = \sum k_{C} v_{C}^{1}, k_{C} \in \mathbb{R}, \emptyset \neq C \subseteq N \}$. The analytic form of $k_{C}$ is represented as $k_{C} = \sum_{T \subseteq C} (-1)^{|C| - |T|} v(T)$.
        \end{definition}

        \textbf{A Set of Cooperative Games.} A cooperative game is usually described as a characteristic function $v: 2^{N} \rightarrow \mathbb{R}_{\geq 0}$, where $N$ is a generic agent set. An arbitrary cooperative game $v$ can be uniquely determined by a set of basis games denoted as $\{v_{C}^{1} \ | \ \emptyset \neq C \subseteq N \}$~\citep{dubey1975uniqueness}, as Definition~\ref{def:game-set} shows. It can be observed that the vector space $G$ is isomorphic to the real coordinate space $\mathbb{R}^{2^{|N|}-1}$. The basis of $G$ includes $2^{|N|}-1$ basis games, so dim $G$ is $2^{|N|} - 1$, equal to the dimension of $\mathbb{R}^{2^{|N|}-1}$. This implies that it is feasible to study cooperative games through the perspective of a real coordinate space. In this paper, we focus on studying the cooperative game $v$ approximation on the set of all agents, so $v(N)$ and $v_{C}^{1}(N)$ are in consideration only, instead of other coalitions $C \neq N$. According to previous work in MARL~\cite{wang2020shapley}, we consider the superadditive game class, which has been shown to fit cooperative MARL (see Appendix~\ref{subsec:superadditive-game-cmarl} for more details). The condition for $G$ restricted to superadditive games is $k_{C} \geq 0$.

        \begin{theorem}[\cite{dubey1975uniqueness}]
        \label{thm:shapley-axioms}
            Shapley value is a unique payoff allocation function, defined on $G$, which satisifes Efficiency, Symmetry and Additivity.
        \end{theorem}
        \textbf{Shapley Value.} Shapley value $\phi: G \rightarrow \mathbb{R}^{|N|}$ is a multidimensional linear transformation defined on the cooperative game space $G$, given that there are $|N|$ agents in total. Each dimension of $\phi$ indicates the payoff allocation to an agent. It has been proved that Shapley value is the unique payoff allocation method under $G$ that satisfies the following axioms: Efficiency, Additivity and Symmetry~\citep{dubey1975uniqueness}, highlighted in Theorem~\ref{thm:shapley-axioms}. Efficiency and Symmetry have been well exploited in multi-agent reinforcement learning~\citep{wang2020shapley,wang2022shaq}, but Additivity still lacks attention. Formally, Additivity means that $\phi(w_{1}+w_{2}) = \phi(w_{1}) + \phi(w_{2})$, for any $w_{1},w_{2} \in G$. If we consider $m$ possible games, then we have the following expression: $\phi(\sum_{i=1}^{m} w_{i}) = \sum_{i=1}^{m} \phi(w_{i})$, where $w_{1}, w_{2}, ..., w_{m} \in G$ and $\sum_{i=1}^{m} w_{i}$ can be seen as a game that is reproduced by the above games. Since the game space $G$ is a vector space, Additivity has been generalized to Linearity in the previous literature~\citep{dubey1975uniqueness,young1985monotonic}, such that $\phi(\sum_{i=1}^{m} \alpha_{i} w_{i}) = \sum_{i=1}^{m} \alpha_{i} \phi(w_{i})$, for $\alpha_{i} \in \mathbb{R}$. Linearity can lead to Additivity, but the reverse is not necessarily correct with no consideration of Homogeneity. In this paper, we will focus on how Linearity can be leveraged to address problems of open multi-agent systems, e.g. NAHT.
    
\section{A Model for N-Agent Ad Hoc Teamwork: State-Space Cooperative Games}
\label{sec:a-model-for-n-agent-ad-hoc-teamwork}
    Fitting the dynamic environment setting in NAHT, we extend the game space $G$ to the state space $\mathcal{S}$ of Markov games~\cite{littman1994markov}, forming a set of state-space cooperative games. In more detail, for each state $s \in \mathcal{S}$ we have a vector space such as $G(s) = \{ w_{s} \ | \ w_{s} = \sum k_{C} v_{C, s}^{1}, k_{\mathcal{C}} \in \mathbb{R}, \emptyset \neq C \subseteq \tilde{N} \}$, which is generated by a set of basis games, such as $\{v_{C,s}^{1} \ | \ \emptyset \neq C \subseteq \tilde{N} \}$. Given a fixed state $s$, the vector space $G(s)$ is still isomorphic to a real coordinate space $\mathbb{R}^{(2^{|\tilde{N}|} - 1)}$, consistent with the structure of $G$ defined over the agent set $\tilde{N}$. As a result, all properties for $G$ also hold for $G(s)$.
    \begin{definition}
    \label{def:ad-hoc-teamwork-in-cooperative-game}
        An n-agent ad hoc teamwork process is described as a set of state-space cooperative games such as $G_{\texttt{NAHT}} = \oplus_{s \in \mathcal{S}} G(s)$, characterizing an team of agents $\tilde{N}$. Each NAHT process belongs to $G_{\texttt{NAHT}}$ is corresponding to varying $\tilde{N}$ with different agent-type composition. Each basis game $v_{C, s}^{1}$ predicts the state value for the scenario at state $s \in \mathcal{S}$, with at least $|C|$ agents of an agent-type composition determined by parameterization of $k_{C}$.
    \end{definition}
    Restricted to the context of NAHT, we describe a known process by a global space denoted by $G_{\texttt{NAHT}} = \oplus_{s \in \mathcal{S}} G(s)$, where $G(s)$ is a subspace with a fixed state $s \in \mathcal{S}$ to capture the composition of agent types and the team size, as described in Definition~\ref{def:ad-hoc-teamwork-in-cooperative-game}, as well as $\oplus$ indicates the direct sum operation. Literally, a subgame for the state $s$, denoted by $v_{s} \in G(s)$, is formed by the basis games $v_{C,s}^{1}$ up to the coalition whose size is equal to the agent group size $|\tilde{N}|$, such that $v_{s} = \sum_{C \subseteq \tilde{N}} k_{C} v_{C,s}^{1}$.
    \begin{wrapfigure}{r}{0.27\textwidth}
      \centering
      \includegraphics[width=0.27\textwidth]{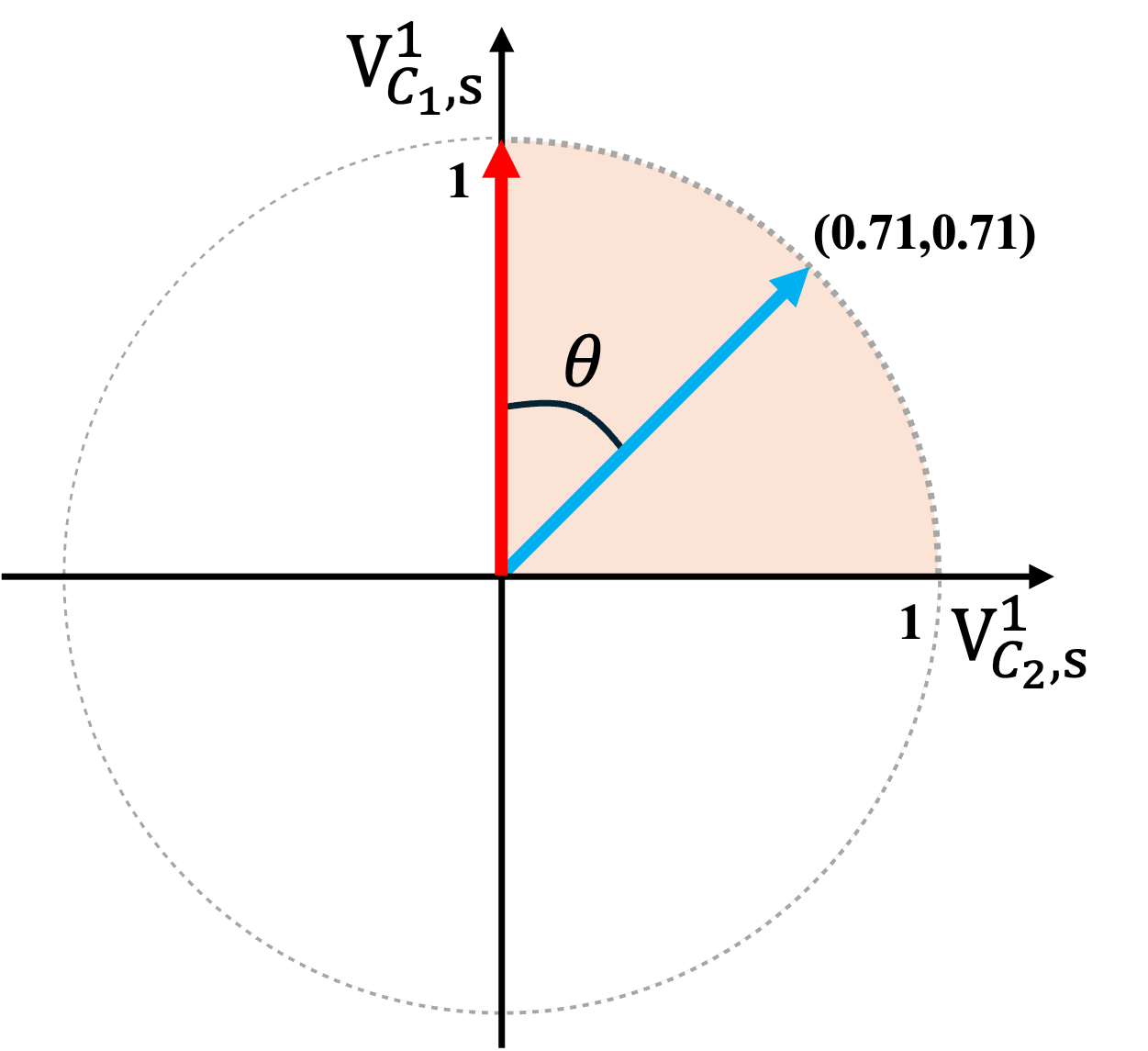}
      \caption{Example of the state-space cooperative game model, described by a spanning set including two vectors with the fixed state $s$. The red and blue arrowed lines are two games, such as $v$ and $w$, whose difference is measured by the $\theta$. The shading area indicates the spectrum of generated cooperative games for $k_{C} \geq 0$ and $\sum k_{C} = 1$, which occupies the first quadrant.}
      \label{fig:state-space-cooperative-game-case}
      \vspace{-20mm}
    \end{wrapfigure}
    $G_{\texttt{NAHT}}$ is isomorphic to $\mathbb{R}^{(2^{|\tilde{N}|} - 1) \times d}$, if we simply consider a finite state space such that $|\mathcal{S}|=d$. Note that the controlled agent set denoted as $\tilde{C}$ is a subset of $\tilde{N}$. Considering uncontrolled agents denoted as $\tilde{U} := \tilde{N} \backslash \tilde{C}$ is necessary for evaluating the controlled agents' policies. If we ignore the uncontrolled agents, the game space with a fixed state will be reduced to a subspace $W \subset \mathbb{R}^{(2^{|\tilde{C}|} - 1)}$, where coordinates not contributing to $\tilde{C}$ are always zeros. The effect of uncontrolled agents is totally ignored. This will induce a bias measured by the angle between two game vectors representing two scenarios. For example, $\theta$ is the angle between two game vectors such that $v_{s} \in W = \text{span}\{v_{\{1\}, s}^{1}\}$ and $w_{s} \in \text{span}\{v_{\{1\}, s}^{1}, v_{\{1,2\}, s}^{1}\} \backslash W$, as shown in Figure~\ref{fig:state-space-cooperative-game-case}. \textbf{As well as providing a visualization of NAHT, the main purpose of this model is for formalizing the use of Shapley value when solving NAHT. More specifically, this supports why coalitions up to $\tilde{N}$ must be considered during the NAHT process, when imposing axiomatic characterization in Section~\ref{sec:shapley-value-for-n-agent-ad-hoc-teamwork}.}

\section{Shapley Value for N-Agent Ad Hoc Teamwork}
\label{sec:shapley-value-for-n-agent-ad-hoc-teamwork}
    As described in Definition~\ref{def:ad-hoc-teamwork-in-cooperative-game}, a NAHT process can be represented by $G_{\texttt{NAHT}}=\oplus_{s \in \mathcal{S}} G(s)$, where each subspace $G(s)$ is generated by a set of basis games with a fixed state $s \in \mathcal{S}$, such as $\{v_{C,s}^{1} \ | \ \emptyset \neq C \subseteq \tilde{N} \}$. Therefore, we can uniquely describe a game $v_{s} \in DG(s)$ by the following formula: $v_{s} = \sum_{\emptyset \neq C \subseteq \tilde{N}} k_{C} v_{C, s}^{1}$, given a fixed state $s$. By Definition~\ref{def:game-set}, we denote $v_{C, s}^{z} = z v_{C, s}^{1}$, for any $z \in \mathbb{R}$, such that if $C \subseteq D$, $v_{C, s}^{z}(D) = z$, otherwise, $v_{C, s}^{z}(D) = 0$. 

    \begin{lemma}
    \label{lemm:subgame-transform}
        Given a fixed state $s \in {\cal S}$ each subgame $v_{s} \in G(s)$, it can be uniquely represented by $v_{s} = \sum_{\emptyset \neq C \subseteq \tilde{N}} \frac{k_{C}}{v_{s}(\tilde{N})} \cdot v_{C,s}^{v_{s}(\tilde{N})}$.
    \end{lemma}
    By Lemma~\ref{lemm:subgame-transform}, we can represent a subgame as follows: $v_{s} = \sum_{\emptyset \neq C \subseteq \tilde{N}} \frac{k_{C}}{v_{s}(\tilde{N})} \cdot v_{C,s}^{v_{s}(\tilde{N})}$. Imposing a linear operator $\phi: G(s) \rightarrow \mathbb{R}^{|\tilde{N}|}$ on the subgame $v_{s}$, we derive the following formula such that
    \begin{equation}
        \phi(v_{s}) = \phi \left( \sum_{\emptyset \neq C \subseteq \tilde{N}} \frac{k_{C}}{v_{s}(\tilde{N})} \cdot v_{C, s}^{v_{s}(\tilde{N})} \right).
    \label{eq:shapley-value-impose}
    \end{equation}
    By the Linearity axiom of $\phi$, we can transform Eq.~\eqref{eq:shapley-value-impose} as follows:
    \begin{equation}
        \phi(v_{s}) = \sum_{\emptyset \neq C \subseteq \tilde{N}} \frac{k_{C}}{v_{s}(\tilde{N})} \cdot \phi \left( v_{C, s}^{v_{s}(\tilde{N})} \right).
    \label{eq:shapley-value-linear}
    \end{equation}

    \subsection{Realization of Efficiency and Additivity Axioms in Markov Games}
    \label{subsec:links-to-rl}
            In this section, we will discuss how the Efficiency and Additivity axioms are realized in Markov games~\cite{littman1994markov}. For brevity, we rewrite $v_{C_{i}, s}^{v_{s}(\tilde{N})}(\tilde{N})$ as $V(C_{i}, s)$ and $v_{s}(\tilde{N})$ as $V(\tilde{N}, s)$, respectively, and $V(C_{1}, s) = V(C_{2}, s) = \cdots = V(C_{m}, s) = V(\tilde{N}, s)$, according to the definition of $v_{C,s}^{z}=z$. 
        
            In the context of Markov games, a subgame evaluating a group of agents $\tilde{N}$, $v_{s_{t}}(\tilde{N}) \in G(s)$, as a state value, can be defined by $\mathbb{E}_{\pi}[\sum_{\tau=0}^{T-t-1} \gamma^{\tau} R_{t+\tau} \ | \ s_{t}]$, following the convention of reinforcement learning~\cite{sutton1988learning}[Chap. 3], where $\pi = \times_{i \in \tilde{N}} \pi_{i}$ is a joint stationary policy of agents belonging to $\tilde{N}$ and $T$ is an episode length. By the Bellman equation with n-step returns, we have the follow-up expansion for $V(\tilde{N}, s_{t})$:
            \begin{equation}
                V(\tilde{N}, s_{t}) = \mathbb{E}_{\pi} \left[ \sum_{\tau=0}^{k-1} \gamma^{\tau} R_{t+\tau} + \gamma^{k} V(\tilde{N}, s_{t+k}) \ \mathlarger{\mathlarger{\mathlarger{\vert}}} \ s_{t} \right].
            \label{eq:game-value-expansion}
            \end{equation}
            Since $V(C_{1}, s) = V(C_{2}, s) = \cdots = V(C_{m}, s) = V(\tilde{N}, s)$, we have the following formula:
            \begin{equation}
                V(C_{i}, s_{t}) = \mathbb{E}_{\pi} \left[ \sum_{\tau=0}^{k-1} \gamma^{\tau} R_{t+\tau} + \gamma^{k} V(\tilde{N}, s_{t+k}) \ \mathlarger{\mathlarger{\mathlarger{\vert}}} \ s_{t} \right].
            \label{eq:game-value-expansion-coalition}
            \end{equation}
            Next, we will discuss how the value of $k$ is uniquely determined for each $V(C_{i}, s_{t})$ under a reasonable assumption. For conciseness, we denote the n-step return as $G_{t:t+n}$ such that
            \begin{equation*}
                G_{t:t+n} = \sum_{\tau=0}^{n-1} \gamma^{\tau} R_{t+\tau} + \gamma^{n} V(\tilde{N}, s_{t+n}).
            \end{equation*}
    
            \begin{assumption}
            \label{assm:correspondence-subgame-n-step-return}
                 For a state $s_{t} \in \mathcal{S}$, the basis games ordered in an ascending sequence of the coalition sizes, $(v_{C_{1}, s_{t}}^{v_{s_{t}}(\tilde{N})}(\tilde{N}), v_{C_{2}, s_{t}}^{v_{s_{t}}(\tilde{N})}(\tilde{N}), ..., v_{C_{m}, s_{t}}^{v_{s_{t}}(\tilde{N})}(\tilde{N}))$, are associated with a sequence of $n$-step returns, $\left( \mathbb{E}_{\pi}[G_{t:t+1} | s_{t}], \mathbb{E}_{\pi}[G_{t:t+2} | s_{t}], ..., \mathbb{E}_{\pi}[G_{t:t+n} | s_{t}] \right)$, such that $\mathbb{E}_{\pi}[G_{t:t+k} | s_{t}] = \sum_{i=p}^{q} v_{C_{i}, s_{t}}^{v_{s_{t}}(\tilde{N})}(\tilde{N})$, for $1 \leq k \leq n$, where $|C_{p-1}| < |C_{p}| = |C_{p+1}| = \cdots = |C_{q}| < |C_{q+1}|$ and $n < m$.
            \end{assumption}
    
            By Assumption~\ref{assm:correspondence-subgame-n-step-return} (\textcolor{black}{see Appendix~\ref{sec:explanation-of-assumptions} for more justification}), we can obtain the formula such that $v_{C_{i}, s_{t}}^{v_{s_{t}}(\tilde{N})}(\tilde{N}) = \frac{1}{q(i)} \cdot \mathbb{E}_{\pi}[G_{t:t+n(i)} | s_{t}]$ due to identical values across $v_{C_{i}, s_{t}}^{v_{s_{t}}(\tilde{N})}(\tilde{N})$, where $q(i)$ denotes the total number of coalitions that have the same size as $C_{i}$, and $n(i)$ denotes the index of the corresponding n-step return $G_{t:t+n(i)}$, with respect to coalition $C_{i}$. Note that coalitions with the identical sizes are assigned the same $n(i)$. By substitution, we rewrite Eq~\eqref{eq:game-value-expansion-coalition} as follows:
            \begin{equation}
                V(C_{i}, s_{t}) = \frac{1}{q(i)} \cdot \mathbb{E}_{\pi} \left[ \sum_{\tau=0}^{n(i)-1} \gamma^{\tau} R_{t+\tau} + \gamma^{n(i)} V(\tilde{N}, s_{t+n(i)}) \ \mathlarger{\mathlarger{\mathlarger{\vert}}} \ s_{t} \right].
            \label{aligned-coalition-value-to-expected-value}
            \end{equation}
    
             We now impose a multidimensional linear operator denoted by $\phi: G(s) \rightarrow \mathbb{R}^{|\tilde{N}|}$ on $V(C_{i}, s_{t})$. By linearity, it follows that:
            \begin{equation}
                \phi(V(C_{i}, s_{t})) = \frac{1}{q(i)} \cdot \mathbb{E}_{\pi} \left[ \sum_{\tau=0}^{n(i)-1} \gamma^{\tau} \phi(R_{t+\tau}) + \gamma^{n(i)} \phi(V(\tilde{N}, s_{t+n(i)})) \ \mathlarger{\mathlarger{\mathlarger{\vert}}} \ s_{t} \right].
            \label{eq:impose-shapley-value}
            \end{equation}

        \subsubsection{Efficiency Axiom}
            \begin{proposition}
            \label{prop:shapley-machine-reward}
                Given the condition that $\sum_{i=1}^{|\tilde{N}|} \phi_{i}(R_{t}) = R_{t}$, the payoff allocation defined on returns $R_{t}$, can be expressed as:
                \begin{equation}
                    \phi_{i}(R_{t}) := R_{t} - \sum_{j \neq i} \left( \ \phi_{j}(V(\tilde{N}, s_{t})) - \gamma \phi_{j}(V(\tilde{N}, s_{t+1})) \ \right).
                \label{eq:individual-reward-approximation}
                \end{equation}
            \end{proposition}
            Note that $\phi(R_{t+\tau})$ in Eq.~\eqref{eq:impose-shapley-value} has not yet been defined. For it to be well defined on the subspace $G(s)$, it is required to represent $\phi(R_{t+\tau})$ in the form of $\phi(V(\tilde{N}, s))$. To achieve this, we now define $\phi(R_{t})$ by introducing the Efficiency axiom. To satisfy the Efficiency axiom such that $\sum_{i=1}^{|\tilde{N}|} \phi_{i}(V(\tilde{N}, s_{t})) = V(\tilde{N}, s_{t})$, it is reasonable to assume that $\sum_{i=1}^{|\tilde{N}|} \phi_{i}(R_{t}) = R_{t}$. In other words, each agent's value expansion can be expressed independently with its own $\phi_{i}(R_{t})$, which will be detailed in the next subsection. Then, the resulting $\phi_{i}(R_{t})$ is detailed in Proposition~\ref{prop:shapley-machine-reward}.

        \subsubsection{Additivity Axiom}
            Recall that $v_{C_{i}, s}^{v_{s}(\tilde{N})}(\tilde{N})$ is denoted by $V(C_{i}, s)$, and $v_{s}(\tilde{N})$ is denoted by $V(\tilde{N}, s)$. Substituting Eq.~\eqref{eq:impose-shapley-value} into the Additivity axiom detailed in Eq.~\eqref{eq:shapley-value-linear} and presuming that there are $m$ n-step returns $G_{t:t+n}$ in total with no loss of generality, we obtain the following formula:
            \begin{equation}
                \phi(V(\tilde{N}, s_{t})) = \sum_{n=1}^{m} \frac{k_{C_{n}}}{V(\tilde{N}, s_{t})} \cdot \mathbb{E}_{\pi} \left[ \sum_{\tau=0}^{n-1} \gamma^{\tau} \phi(R_{t+\tau}) + \gamma^{n} \phi(V(\tilde{N}, s_{t+n})) \ \mathlarger{\mathlarger{\mathlarger{\vert}}} \ s_{t} \right].
            \end{equation}
    
            By linearity of the expectation operator, we obtain that:
            \begin{equation}
                \phi(V(\tilde{N}, s_{t})) = \mathbb{E}_{\pi} \left[ \sum_{n=1}^{m} \frac{k_{C_{n}}}{V(\tilde{N}, s_{t})} \cdot \sum_{\tau=0}^{n-1} \gamma^{\tau} \phi(R_{t+\tau}) + \gamma^{n} \phi(V(\tilde{N}, s_{t+n})) \ \mathlarger{\mathlarger{\mathlarger{\vert}}} \ s_{t} \right],
            \label{eq:shapley-value-expected-game}
            \end{equation}
            where the term in the $\mathbb{E}_{\pi}[\cdot]$ can be seen as a form of truncated $\lambda$-return~\cite{sutton1988learning}, with $\frac{k_{C_{n}}}{V(\tilde{N}, s)}$ as a weight, generated by a distribution parameterized by $\lambda$. We will specify this in the next section.

    \subsection{Implementation of Additivity, Efficiency and Symmetry in N-Agent Ad Hoc Teamwork}
    \label{subsec:implementation-in-marl}
        In this section, we will discuss how the conditions derived above can be implemented in NAHT, to fulfill the axiomatic characterization: Additivity, Efficiency and Symmetry. From the perspective of various composition of the axiomatic characterization, we propose a new algorithm named Shapley Machine, and interpret the existing algorithm, POAM~\cite{wang2025n}, in Section~\ref{subsec:summary-of-proposed-algorithms}.
        
        \subsubsection{Additivity Axiom}
            First, we will learn $\phi_{i}(V(\tilde{N}, s_{t}))$ directly for each controlled agent $i \in \tilde{C}$. For brevity, we can represent each $\phi_{i}(V(\tilde{N}, s_{t}))$ as $V_{i}(s_{t})$, and $\phi_{i}(R_{t})$ as $R_{t,i}$, respectively. By sampling trajectories under a joint policy $\pi$, we derive the following TD error $\delta_{t,i}$ for each agent $i$ such that
            \begin{equation}
                \delta_{t,i} = \sum_{n=1}^{m} \frac{k_{C_{n}}}{V(\tilde{N}, s)} \cdot \sum_{\tau=0}^{n-1} \left( \gamma^{\tau} R_{t+\tau,i} + \gamma^{n} V_{i}(s_{t+n}) \right) - V_{i}(s_{t}).
            \label{eq:td-error-general}
            \end{equation}
            \begin{proposition}
                For the class of superadditive games formed by a set of basis games $\{v^1_{C,s} | \emptyset \neq C \subseteq \tilde{N} \}$, it holds that $k_{C} \geq 0$, for all $\emptyset \neq C \subseteq \tilde{N}$.
            \label{prop:positive-Kc}
            \end{proposition}
            For brevity, we denote $k_{C_{n}}' = \frac{k_{C_{n}}}{V(\tilde{N}, s)}$. Since in definition $V(\tilde{N}, s) > 0$, we can conclude that $k_{C_{n}}' > 0$ is a sufficient condition for $V(\tilde{N}, s)$ belonging to the superadditive game class, according to Proposition~\ref{prop:positive-Kc}. In other words, $k_{C_{n}}' > 0$ facilitates cooperation within a group of agents, following the relationship between superadditive games and cooperative MARL described in~\cite{wang2020shapley} (see Appendix~\ref{subsec:superadditive-game-cmarl} for details). 
            
            Following the convention of TD($\lambda$) in RL~\cite{sutton1988learning}, the values of $(k_{C_{1}}', k_{C_{2}}', \cdots, k_{C_{m}}')$ are generated using a geometric distribution $P_{\lambda}$ with the parameter $0 < \lambda < 1$, such that $k_{C_{i}}' = P_{\lambda}(i)$, resulting $((1-\lambda), (1-\lambda)\lambda, \cdots, (1-\lambda)\lambda^{m-1}, \lambda^{m})$. With this condition, Eq.~\eqref{eq:td-error-general} becomes the TD error of the well-known truncated TD($\lambda$) prediction, shortened as TTD($\lambda$)~\cite{cichosz1994truncating}. 
        
        \subsubsection{Efficiency Axiom}
            By substitution into Eq.~\eqref{eq:individual-reward-approximation}, the $R_{t+\tau,i}$ in Eq.~\eqref{eq:td-error-general} can be expressed as 
            \begin{equation}
                R_{t,i} = R_{t} - \sum_{j \neq i} \left( V_{j}(s_{t}) - \gamma V_{j}(s_{t+1}) \right).
            \label{eq:shape-rewards}
            \end{equation}
            \begin{theorem}[\cite{oliehoek2016concise}]
            \label{thm:factored-dec-pomdp}
                Given an additively factored immediate reward function, for any timestep $t$ there is a factorization scheme of the Dynamic Bayesian Network (describing the transition function), such that the value of a finite-horizon factored Dec-POMDP is decomposable across agents.
            \end{theorem}
            $R_{t} = \sum_{i=1}^{|\tilde{N}|} \phi_{i}(R_{t})$ to derive the Efficiency axiom is also referred to as additively factored immediate reward function. Theorem~\ref{thm:factored-dec-pomdp} indicates that the Efficiency axiom such that $V(\tilde{N}, s_{t}) = \sum_{i=1}^{|\tilde{N}|} \phi_{i}(V(\tilde{N}, s_{t}))$ always exists in the context of Dec-POMDP~\cite{bernstein2002complexity}, given a factorization scheme of the transition function. However, we may optionally add one additional condition to implicitly facilitate searching for a factorization scheme during training such that
            \begin{equation}
                \sum_{i=1}^{|\tilde{N}|} \phi_{i}(V(\tilde{N}, s_{t})) = V(\tilde{N}, s_{t}).
            \label{eq:efficiency}
            \end{equation}
        
        \subsubsection{Symmetry Axiom}
            In practice, an agent is only able to receive an observation, in theory generated from a function of the state in Dec-POMDPs, a partially observable Markov game. This leads to that an agent is required to infer the state, as an individual hidden state, by the history of observations, e.g., using recurrent neural networks (RNNs), e.g. GRUs~\cite{chung2014empirical}. Owing to the feature of NAHT, each agent has to infer other agents' characteristics, which are used as inputs in addition to the individual hidden state, for both individual value functions and policies. This has been accomplished by the previous work, POAM~\citep{wang2025n}, a variant of IPPO~\citep{de2020independent} with an encoder-decoder module to infer other agents' characteristics, as a joint embedding vector. The value functions (critics) with the concatenation of individual hidden state and the joint embedding. Recall that the Symmetry axiom states that two agents who contribute equally to all coalitions excluding themselves should receive equal payoff allocations~\cite{chalkiadakis2011computational}[Chap. 2]. In the context of NAHT, this axiom applies when agents $i$ and $j$ share identical agent types. Since their types are the same, the two agents are exchangeable—from the perspective of any third agent, $i$ and $j$ are indistinguishable when included as teammates. Similarly, from each agent’s own perspective, the other agent appears identical within the teammate set. As a result, when interacting with the same teammates, agents $i$ and $j$ will exhibit identical behavioral trajectories, leading to identical inferred individual hidden states. Furthermore, since the joint embedding vector of teammates is computed based on their behavioral trajectories, and the teammates are fixed while $i$ and $j$ behave identically, both agents will infer the same joint teammate representation. Consequently, their expected payoff allocations, $V_i$ and $V_j$, are equal, thereby satisfying the Symmetry axiom.
            
    \subsection{Discussion on Proposed Algorithms and POAM}
    \label{subsec:summary-of-proposed-algorithms}
        To summarize, POAM uses TD($\lambda$) as the target values to update critics, and the global reward $R_{t}$ for each agent's policy optimization. Our algorithm named Shapley Machine is implemented based on POAM, by simply replacing the global reward $R_{t}$ with each agent's $R_{t,i}$ defined in Eq.~\eqref{eq:shape-rewards}, using TTD($\lambda$) with the horizon proportional to the number of possible agents, and adding Eq.~\eqref{eq:efficiency} as a regularization term.
        \begin{proposition}
        \label{prop:shapley-machine}
            Shapley Machine is an algorithm that fulfills Efficiency, Additivity and Symmetry, so it learns $V_{i}$ as Shapley values for dynamic scenarios.
        \end{proposition} 
        
        Note that Shapley Machine is forced to comply with Efficiency, Additivity and Symmetry during learning, so the resulting $V_{i}$ has to be Shapley value, as highlighted in Proposition~\ref{prop:shapley-machine}.
        \begin{remark}
            Banzhaf Machine is an algorithm which generates $V_{i}$, fulfilling Additivity and Symmetry.
        \end{remark}
        Following the principle of designing algorithms by complying with axiomatic characterization, we also propose an algorithm that satisfies conditions corresponding to Additivity and Symmetry. As the resulting $V_{i}$ is likely to be Banzhaf index~\cite{banzhaf1964weighted},~\footnote{It is known that Banzhaf index is not the only payoff allocation scheme that satisfies Additivity and Symmetry.} the algorithm is referred to as Banzhaf Machine.
        \begin{remark}
            POAM is a Banzhaf Machine that learns Banzhaf indices for dynamic scenarios. 
        \end{remark}
        According to the taxonomy of our theory, POAM also satisfies Additivity and Symmetry, so it is a Banzhaf Machine. However, it uses TD($\lambda$) rather than TTD($\lambda$), which means that it may overparameterize the value function in some situations where there are few agents \textcolor{black}{(justified by the MPE case in Figure~\ref{fig:shapley-machine-m})}. 

\section{Experiments}
\label{sec:experiments}
    In experiments, we evaluate our proposed algorithm on MPE and SMAC tasks designed by~\cite{wang2025n}, compared with the only baseline algorithm POAM for NAHT. Also, we demonstrate the importance of basis games, a core concept in our theory, as an aid to setting the number of components considered in TTD($\lambda$). All experimental results are run with 5 random seeds and demonstrated in means with 95\% confidence intervals. The details of baselines, the implementation of our algorithm and further experimental settings are given in \textcolor{black}{Appendix~\ref{sec:experimental-details}}.
    \begin{figure}[ht!]
        \centering
        \includegraphics[width=\linewidth]{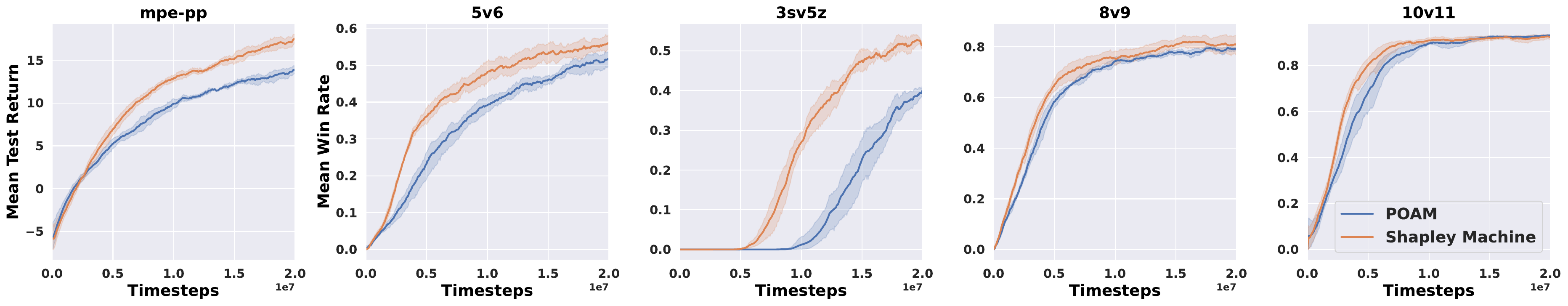}
        \caption{Testing evaluation during the training procedure.}
    \label{fig:test-return}
    \end{figure}
    \begin{figure}[ht!]
        \centering
        \includegraphics[width=\linewidth]{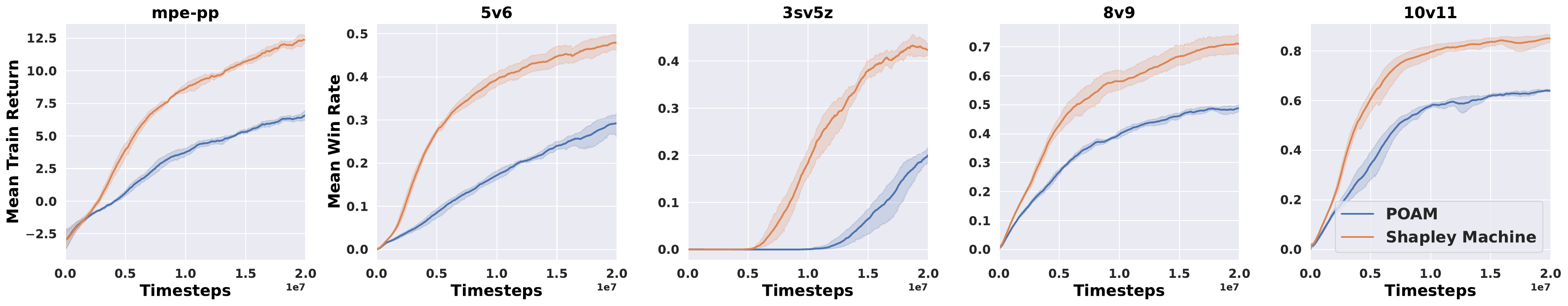}
        \caption{Training evaluation during the training procedure.}
    \label{fig:train-return}
    \end{figure}
    \begin{figure}[ht!]
        \centering
        \includegraphics[width=\linewidth]{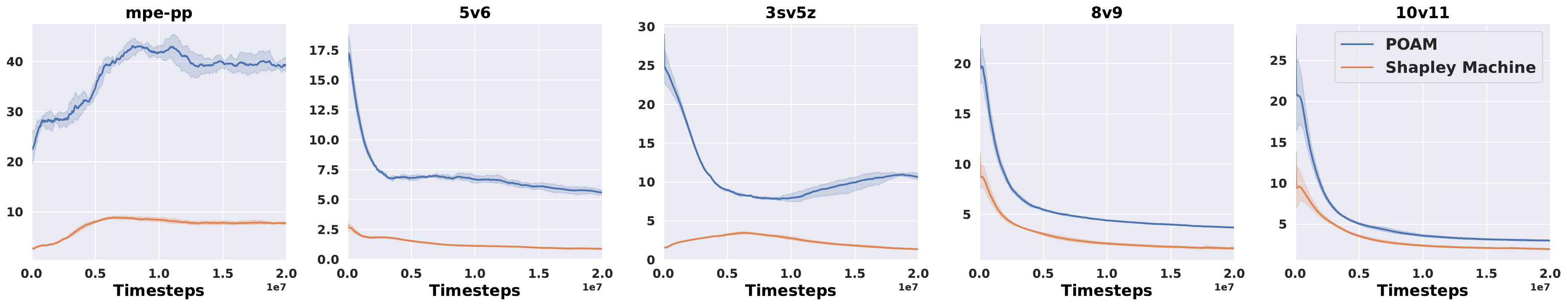}
        \caption{Critic loss record across all scenarios during the training procedure.}
    \label{fig:train-critic-loss}
    \end{figure}
    
    \subsection{Main Results}
    \label{subsec:main-result}
        We first show the general performance of Shapley Machine, compared with POAM. It can be observed in Figure~\ref{fig:test-return} that in three scenarios (MPE, 3sv5z and 5v6), Shapley Machine outperforms POAM in testing. In contrast, Shapley Machine can only match the performance of POAM in both 8v9 and 10v11 during testing. On the other hand, Figure~\ref{fig:train-return} shows that Shapley Machine converges faster during training. A possible reason is that the structure imposed by cooperative games facilitates the predicting of each agent's payoff. This can also be verified by Figure~\ref{fig:train-critic-loss}, in which we observe that the critic loss of Shapley Machine is consistently lower than POAM. Similar phenomenon can be observed for the entropy of policies \textcolor{black}{(see Appendix~\ref{subsec:further-evidence-on-the-benefit-of-structured-algorithm})}. This implies that agents trained with Shapley Machine receive more accurate credit during learning, which justifies the correctness of our theory and the effectiveness of Shapley Machine.

    \begin{figure}[htbp]
      \centering
      \begin{subfigure}[b]{0.45\textwidth}
        \centering
        \includegraphics[width=\linewidth]{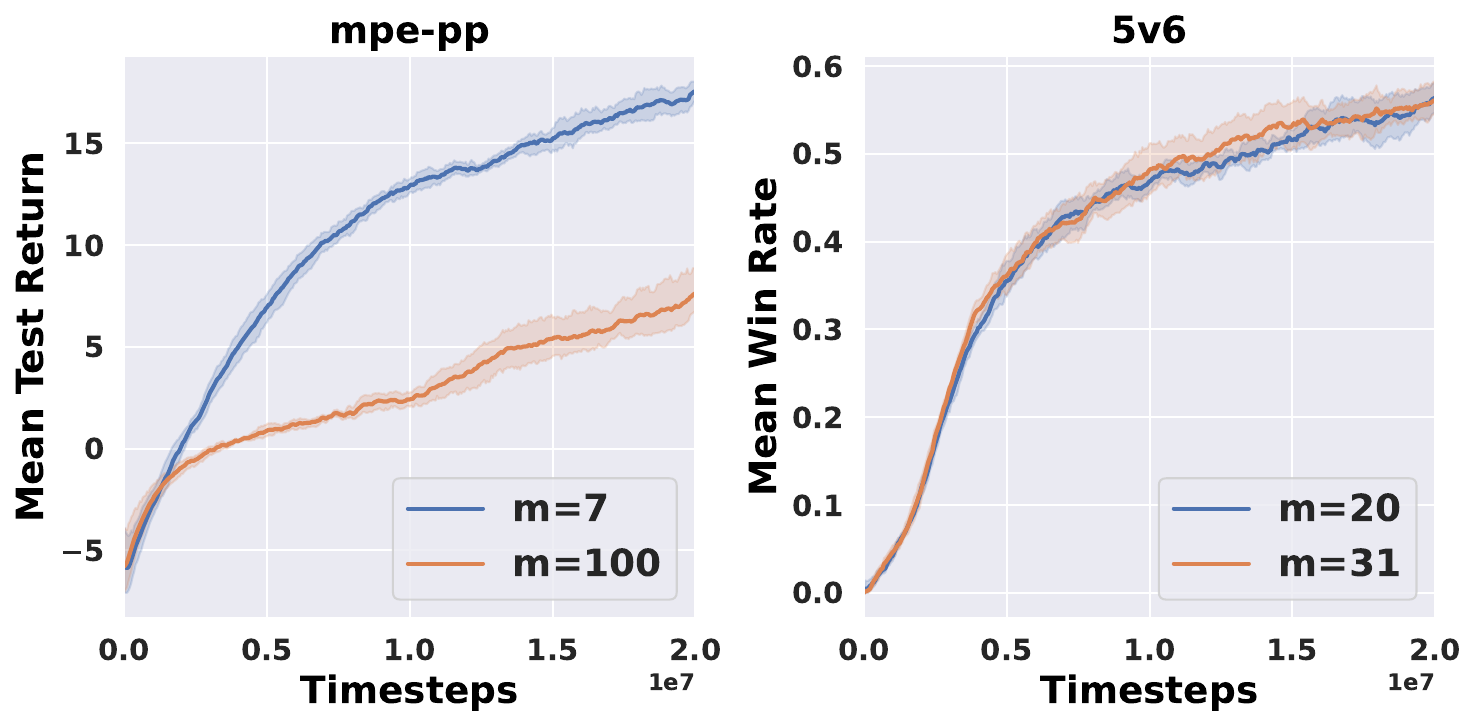}
        \caption{Shapley Machine with various $m$ values.}
        \label{fig:shapley-machine-m}
      \end{subfigure}
      \hfill
      \begin{subfigure}[b]{0.45\textwidth}
        \centering
        \includegraphics[width=\linewidth]{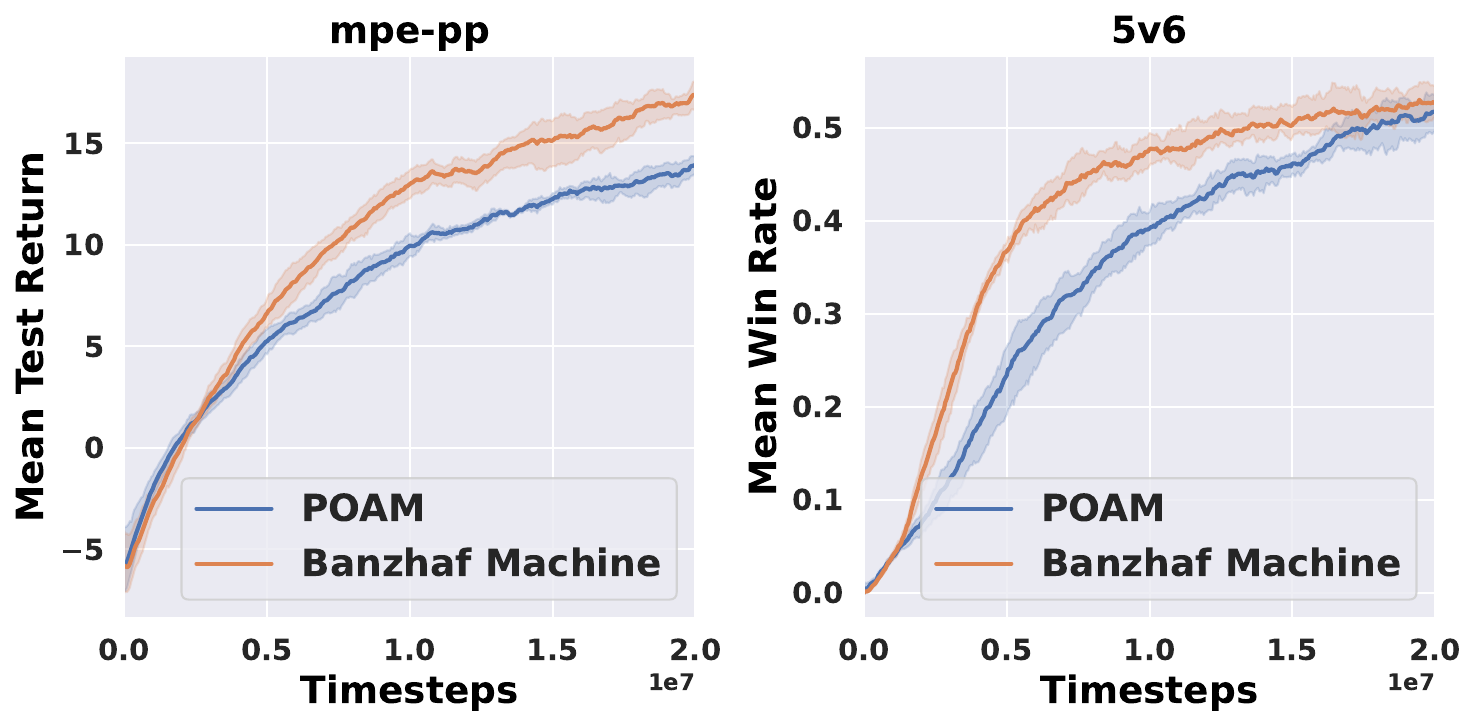}
        \caption{POAM vs. Banzhaf Machine.}
        \label{fig:poam-vs-banzhaf-machine}
      \end{subfigure}
    \caption{Proof-of-concept experiments for verifying the necessity of basis games. In MPE, the expected number of basis games $m$ is 7, while it is 31 in 5v6. Without the structure of our theory, there is no evidence to decide the value of $m$. Instead, TD($\lambda$) can only be used to calculate n-step returns up to the episode length. The episode length of MPE is 100.}
    \label{fig:basis-games}
    \end{figure}
    \subsection{Proof of Concepts in Our Theory}
        \textbf{Horizon $m$ of TTD($\lambda$) Equal to the Number of Basis Games.} We now investigate whether the horizon of TTD($\lambda$), denoted by $m$, is associated with the number of basis games according to our theory. In Figure~\ref{fig:shapley-machine-m}, we can observe that the performance of Shapley Machine varies among various $m$. Recall that in MPE the total number of agents is 3, so in our theory the optimal value of $m$ is 7,~\footnote{Recall that the number of basis games is equal to the number of non-empty coalitions of an agent set.} which is consistent with the result. To eliminate the possibility that this result is correlated with the Efficiency axiom of Shapley Machine, we also conduct an experiment on comparing POAM~\footnote{POAM uses TD($\lambda$), equivalently, $m$ is equal to the episode length.} and its variant with the optimal $m$ as 7, referred to as Banzhaf Machine, in Figure~\ref{fig:poam-vs-banzhaf-machine}. It can be observed that the optimal $m$ is still the number of basis games, conforming to our theory. This result reveals a general empirical finding that \textbf{TTD($\lambda$) tends to perform better than TD($\lambda$) in NAHT, with the horizon estimated by the number of agents}.
        
        \textbf{Approximating $m$ for Large-Scale Scenarios.} In scenarios with a large number of agents, it is almost impossible to consider all basis games (the number of which could be much larger than the episode length). To explore if the number of a subset of basis games can work for these scenarios, we further conduct case studies on 5v6. As shown in Figure~\ref{fig:shapley-machine-m}, using $m=20$ to approximate the theoretical number of basis games given by $m=31$ will not reduce the performance. In additional experiments (\textcolor{black}{see Appendix~\ref{subsec:empirical-m-values-and-number-of-agents}}), we observe that the empirical optimal value of $m$ is still positively correlated to the number of agents, this providing a law to seek $m$ manually for the large-scale scenarios, such as 8v9 and 10v11. Nevertheless, establishing an empirical function that maps the number of agents to $m$ will be addressed in future work. In this paper, we temporarily use the episode lengths as empirical values for $m$ in the 8v9 and 10v11 scenarios, which coincidentally correspond to the ascending order of their coalition sizes. This observation naturally raises an open question: \textbf{Is the episode length of a task related to the maximum number of agents involved?} Furthermore, it is empirically acknowledged that episode length often correlates with task difficulty and can facilitate more effective exploration. From this perspective, a second question arises: \textbf{Is the maximum number of agents indicative of task difficulty?} Both questions warrant deeper investigation in future work.

\section{Related Work}
\label{sec:related-work}
    \subsection{TD($\lambda$) in Reinforcement Learning and Multi-Agent Reinforcement Learning}
        Related work on temporal-difference learning with eligibility traces centers on the seminal TD($\lambda$) algorithm proposed by \cite{sutton1988learning}, which unified one-step TD and Monte Carlo methods via a trace-decay parameter $\lambda$ and introduced the forward‐ and backward‐view equivalence for efficient multi-step credit assignment. Subsequent theoretical analyses \cite{dayan1992convergence,dayan1994td,tsitsiklis1996analysis} established convergence guarantees and characterized the bias–variance trade-off inherent in $\lambda$-returns, while extensions by \cite{watkins1989learning,peng1994incremental} adapted eligibility traces for off-policy control, Q($\lambda$). Advances in function approximation, including true online TD($\lambda$) \cite{seijen2014true} and gradient-TD methods, further broadened applicability to large-scale and nonlinear settings, inspiring n-step return techniques such as generalized advantage estimation (GAE) \cite{schulman2015high} for policy optimization. In multi-agent reinforcement learning, algorithms of PPO family such as MAPPO \cite{yu2022surprising} and IPPO \cite{de2020independent}, applied GAE to optimize policies and thus used TD($\lambda$) as the target values to train critics. This paper sheds light on the link between truncated TD($\lambda$), shortened as TTD($\lambda$) in NAHT (a generalization of MARL) and Shapley value in theory. Furthermore, it offers the insight that the number of n-step return components in TTD($\lambda$) should be larger as the number of agents increases. This theoretically insightful conclusion has the potential to be applied to a broader range of multi-agent reinforcement learning tasks.

    \subsection{Theoretical Models for Ad Hoc Teamwork}
        We now discuss theoretical models describing ad hoc teamwork (AHT). \cite{brafman1996partially} pioneered research into AHT by investigating repeated matrix games involving a single teammate. Subsequent studies expanded this framework to scenarios with multiple teammates, notably by \cite{agmon2012leading}. Later, \cite{agmon2014modeling} further relaxed earlier assumptions by allowing teammates' policies to be selected from a known set. Stone et al. \cite{stone2010teach} initially formalized AHT through collaborative multi-armed bandits, albeit under notable assumptions such as prior knowledge of teammates' policies and environmental conditions. \cite{albrecht2013game} advanced this field significantly by introducing the stochastic Bayesian game (SBG), the first comprehensive theoretical framework accommodating dynamic environments and unknown teammate behaviors in AHT. Building on SBG, \cite{rahman2021towards} proposed the open stochastic Bayesian game (OSBG), addressing open ad hoc teamwork (OAHT). \cite{ZintgrafDCWH21} modeled AHT through interactive Bayesian reinforcement learning (IBRL) within Markov games, specifically targeting non-stationary teammate policies within single episodes. \cite{xie2021learning} introduced the hidden parameter Markov decision process (HiP-MDP) to handle situations where teammates' policies vary across episodes but remain stationary during individual episodes. Most recently, \cite{wang2024open} extended OSBG by incorporating principles from cooperative game theory, introducing the open stochastic Bayesian coalitional affinity game (OSBG-CAG), which theoretically justifies a graph-based representation for joint Q-value functions and includes rigorous convergence proofs for Q-learning algorithms in open team scenarios. This paper stands on the view of cooperative game theory, establishing a state-space cooperative game model to deal with the varying number of controlled agents and uncontrolled agents. Furthermore, this model serves as the foundation for the Shapley value-based approach to NAHT.
    
    \subsection{Shapley Value in Multi-Agent Reinforcement Learning}
        Related work mainly focused on developing the theory of Shapley value in MARL, and incorporated Shapley value~\cite{shapley1953value} into MARL algorithms. Early studies \cite{wang2020shapley} incorporated Shapley value into credit assignment scheme in a principled way and proposed an algorithm named SQDDPG, underpinned by the equivalence between cooperative-game theoretical models in dynamic scenarios and the shared reward Markov games. \cite{wang2022shaq} further improved the theory by proving the existence of Shapley value, and proposed an algorithm named SHAQ, an algorithm with promising convergence to find an optimal joint policy. It also shed light on the relationship between Shapley value and the relevant value decomposition and credit assignment approaches for MARL. \cite{li2021shapley} improved the stability of SQDDPG and proposed an algorithm named Shapley Counterfactual Credits. \cite{han2021multiagent} incorporated Shapley value into the model-based PPO as each agent's advantage value, and estimated coalition values by trained transition and reward models. \cite{chai2024aligning} continued the idea from \cite{han2021multiagent}, by replacing the the transition and reward models with more powerful world models. \cite{xue2022efficient} incorporated Shapley value into multi-agent communication, as a criterion for forming communication between agents. This paper primarily focuses on learning the Shapley value to address NAHT, a generalized paradigm of MARL characterized by an uncertain number of controlled and uncontrolled agents. Furthermore, the Shapley value is learned by satisfying its axioms, rather than by constructing an explicit form as done in previous work. This may help overcome the notorious instability issues in learning Shapley value-based algorithms.
        
\section{Conclusion}
\label{sec:conclusion}
    \textbf{Summary.} This paper addresses NAHT through the lens of cooperative game theory. Specifically, it first establishes a state-space cooperative game model, in the form of a vector space that describes NAHT processes, extended from the cooperative game space from cooperative game theory. Relying on the axioms defined on the state-space cooperative game model (Additivity, Efficiency and Symmetry), we derive an algorithm named Shapley Machine that is provable to uniquely learn Shapley value. Our approach is different from the traditional usage of Shapley value in MARL, in which value functions are usually constructed by Shapley value's explicit form. Rather, we show that Shapley values are subtly related to the TD($\lambda$), a well-known algorithm in RL dealing with bias-variance tradeoff in TD prediction.

    \textbf{Discussion, Limitation and Future Work.} In addition to proposing a new algorithm, a key result of this work is to shed light on designing multi-agent learning algorithms through axiomatic characterization. This not only helps to categorize existing algorithms (e.g., POAM can be categorized as an algorithm using Banzhaf index as a credit assignment scheme), but aids in designing new algorithms on the basis of clearly interpretable formal principles (e.g., Shapley Machine is designed by fulfilling three axioms). This will decrease the possibility of reinvention of algorithms by clarifying the relationship between the newly designed algorithm and previous ones. This principle can be extended to the future NAHT and MARL algorithm design, accelerating the pace of research into trustworthy AI. On the other hand, this paper unveils the potential to investigate the TTD($\lambda$) technique, though the hyperparameters like $\lambda$ and horizon $m$ can only be manually tuned. This is one of limitations of this work. In future work, the relation between $\lambda$ and task categories, as well as the relation between $m$ and the number of agents, can be further investigated. One possible solution is adaptively learning these hyperparameters through real-time feedback during training. Moreover, in our theory the weightings in TTD($\lambda$) (i.e. $k_{C}'$ in the context of our theory) can be negative (see Figure~\ref{fig:state-space-cooperative-game-case}), which implies that the class of games is broader than superadditive games from the perspective of cooperative game theory. This opens a new route towards designing weightings in a richer function class, rather than merely assuming a geometric distribution following the convention of RL. Finally, it has been shown in experiments that Shapley Machine does not work well for large-scale environments, which we will explore in the future work, especially establishing an empirical function that maps the number of agents to $m$.

\section*{Acknowledgement}
    This project is supported by the Engineering and Physical Sciences Research Council [Grant Ref: EP/Y028732/1]. The computational resources are supported by CSC -- IT Center for Science LTD., Finland.

\bibliography{main}
\bibliographystyle{unsrt}

\appendix
        
\section{Additional Background}
\label{sec:additional-background}
    \subsection{$\lambda$-Return and TD($\lambda$)}
    \label{subsec:lambda-return}
        We now introduce an extension of return named $\lambda$-return. Mathematically, a $\lambda$-return $G_{t}^{\lambda}$ for infinite-horizon cases can be expressed as follows:
        \begin{equation}
            \begin{split}
                G_{t:t+n} = R_{t} + \gamma R_{t+1} + \cdots + \gamma^{n-1} R_{t+n} + \gamma^{n} V(s_{t+n}),\\
                G_{t}^{\lambda} = (1 - \lambda) \sum_{n=0}^{\infty} \lambda^{n} G_{t:t+n}.
            \end{split} 
        \end{equation}
        Similarly, the $\lambda$-return for finite-horizon cases can be expressed as follows:
        \begin{equation}
            G_{t}^{\lambda} = (1 - \lambda) \sum_{n=1}^{T-t} \lambda^{n} G_{t:t+n} + \lambda^{T-t} G_{t:T}.
        \end{equation}
        Note that if $\lambda=1$, updating according to the $\lambda$-return is a Monte Carlo algorithm. In contrast, if $\lambda=0$, then the $\lambda$-return reduces to $G_{t:t+1}$, the one-step return. Monte Carlo algorithm is known as its high variance but low bias, while one-step return is known as its low variance but high bias. For this reason, the change of $\lambda$ can be seen as a tradeoff between bias and variance. The TD prediction using $\lambda$-return as the target value is referred to as TD($\lambda$)~\cite{sutton1988learning}.

        A variant of TD($\lambda$) by shaping $G_{t}^{\lambda}$ with a designated horizon $h$ is referred to as truncated TD($\lambda$), shortened as TTD($\lambda$)~\cite{cichosz1994truncating}. Mathematically, $G_{t}^{\lambda}$ with a designated horizon $h$ can be formulated as follows:
        \begin{equation}
            G_{t}^{\lambda} = (1 - \lambda) \sum_{n=0}^{h-t} \lambda^{n} G_{t:t+n} + \lambda^{h-t} G_{t:h}.
        \end{equation}

        Note that in the convention of RL, the horizon $h$ is defined as the time difference from the starting point of an episode, denoted by $t=0$. In this paper, to simply the presentation of our work, we redefine the horizon as the time difference from each timestep $t$, denoted as $m$, such that $m=h-t$. With the new definition of horizon, we can rewrite the above formula of $G_{t}^{\lambda}$ as follows:
        \begin{equation}
            G_{t}^{\lambda} = (1 - \lambda) \sum_{n=1}^{m} \lambda^{n} G_{t:t+n} + \lambda^{m} G_{t:t+m}.
        \end{equation}
        
    \subsection{Superadditive Game and Cooperative Multi-Agent Reinforcement Learning}
    \label{subsec:superadditive-game-cmarl}
        Superadditive game is a subclass of characteristic games that satisfies the additional condition $v(C) + v(D) \geq v(C \cup D)$, for any two distinct coalitions $C,D \subseteq N$ with $C \cap D = \emptyset$. Previous work in MARL \cite{wang2020shapley} has shown the equivalence on objective functions between state-based superadditive games with action space and team reward Markov games with coalition structures.~\footnote{In \cite{wang2020shapley}, state-based convex games with action space were proved to be equivalent to team reward Markov games. However, setting $C \cap D = \emptyset$ can reduce a convex game to a superadditive game.} Intuitively, Shapley value is a solution to the stability of a team formation, resembling cooperation. Since NAHT is a generalization of MARL, this conclusion still valid here in principle. It is already proved that Shapley value exists in superadditive games \cite{shapley1953value} and each game instance in superadditive games can be uniquely represented as $\sum k_{C} v_{C}^{1}$, where $k_{C} \geq 0$~\citep{dubey1975uniqueness}. In cooperative game theory, any game can be transformed to a superadditive game if it is equipped with a superadditive cover such that $v^{*}(C) = \max_{CS(C)} \sum_{D \in CS(C)} v(D)$~\cite{greco2011complexity}. As a result, it will not lose generality if we consider the set of cooperative games as superadditive games in this paper.
        
    \subsection{Policy Optimization with Agent Modeling (POAM)}
    \label{subsec:poam}
        In general, POAM is an algorithm built upon IPPO~\citep{de2020independent}. The primary difference between POAM and IPPO is as follows: For each agent's policy and value, in addition to its observation used in IPPO, POAM takes as input an embedding vector to predict other agents' potential behaviors. More specifically, this embedding vector is trained by an encoder-decoder structure, wherein the encoder is modeled as RNNs and decoder is modeled as MLPs. For brevity, $-i$ denotes the set of all agents excluding agent $i$. Let $h_{t,i} =\{o_{k,i}, a_{k-1,i}\}_{k=1}^{t}$ denote agent $i$'s history of observations and actions up to timestep $t$ and $e_{t,i} \in \mathbb{R}^{n}$ denote the resulting embedding vector of $n$ dimensions. The encoder parameterized by $\theta^{e}$ is defined as $f_{\theta^{e}}^{enc}: \mathcal{H}_{i} \rightarrow \mathbb{R}^{n}$. The embedding vector is decoded by two decoder networks: the observation decoder $f_{\theta^{o}}^{dec}: \mathbb{R}^{n} \rightarrow \mathcal{O}_{-i}$, and the action decoder $f_{\theta^{a}}^{dec}: \mathbb{R}^{n} \rightarrow \Delta(\mathcal{A}_{-i})$. The decider networks are respectively trained to predict the observations and actions of all other agents on the team at timestep $t$, $o_{t,-i}$ and $a_{t,-i}$, to encourage $e_{t,i}$ to contain information about collective behaviors corresponding to $h_{t,i}$. While the observation decoder directly predicts the observed $-i$'s observations, the action decoder predicts the parameters of a probability distribution over the $-i$'s actions $\pi_{-i}(a_{t,-i}; f_{\theta^{a}}^{dec}(f_{\theta^{e}}^{enc}(h_{t,i})))$. As we consider continuous observations and discrete actions, the loss function with using Categorical distribution to model $\pi_{-i}$ is as follows:
        \begin{equation}
            L_{\theta^{e}, \theta^{o}, \theta^{a}}(h_{t,i}, o_{t,-i}, a_{t,-i}) = || f_{\theta^{o}}^{dec}(f_{\theta^{e}}^{enc}(h_{t,i})) - o_{t,-i} ||^{2} - \log \pi_{-i}(a_{t,-i}; f_{\theta^{a}}^{dec}(f_{\theta^{e}}^{enc}(h_{t,i}))).
        \label{eq:loss-encoder-decoder}
        \end{equation}
        
        IPPO employed generalized advantage estimation (GAE)~\citep{schulman2015high} to form the policy gradient, and the value function used to form GAE is trained by the following loss function:
        \begin{equation}
            L_{\theta^{c}}(h_{t,i}) = \frac{1}{2} \left( V_{i}^{\theta^{c}}(h_{t,i}, f_{\theta^{e}}^{enc}(h_{t,i})) - \hat{V}_{t,i} \right)^{2},
        \label{eq:loss-critic}
        \end{equation}
        where $\hat{V}_{t,i}$ is the finite-horizon TD($\lambda$) return, for which the horizon is up to the episode length $T$. Each agent $i$'s GAE estimation is as follows:
        \begin{equation}
            A_{t,i} = \sum_{l=0}^{T} (\gamma \lambda)^{l} \delta_{t+l,i},
        \end{equation}
        where $\delta_{t,i} = R_{t} + \gamma V_{i}^{\theta^{c}}(h_{t+1,i}, f_{\theta^{e}}^{enc}(h_{t+1,i})) - V_{i}^{\theta^{c}}(h_{t,i}, f_{\theta^{e}}^{enc}(h_{t,i}))$. Based on the GAE defined above, the policy optimization loss is defined as:
        \begin{equation}
            L_{\theta}(h_{t,i}, a_{t,i}) = \min \left\{ \frac{\pi_{\theta}(a_{t,i} | h_{t,i})}{\pi_{\theta_{old}}(a_{t,i} | h_{t,i})} A_{t,i}, \text{clip}\left( \frac{\pi_{\theta}(a_{t,i} | h_{t,i})}{\pi_{\theta_{old}}(a_{t,i} | h_{t,i})}, 1 - \epsilon, 1 + \epsilon \right) A_{t,i}\right\}.
        \label{eq:loss-policy}
        \end{equation}

\section{Justification of Assumptions}
\label{sec:explanation-of-assumptions}
    \begin{assumption}
         For a state $s_{t} \in \mathcal{S}$, the basis games ordered in an ascending sequence of the coalition sizes, $(v_{C_{1}, s_{t}}^{v_{s_{t}}(\tilde{N})}(\tilde{N}), v_{C_{2}, s_{t}}^{v_{s_{t}}(\tilde{N})}(\tilde{N}), ..., v_{C_{m}, s_{t}}^{v_{s_{t}}(\tilde{N})}(\tilde{N}))$, are associated with a sequence of $n$-step returns, $\left( \mathbb{E}_{\pi}[G_{t:t+1} | s_{t}], \mathbb{E}_{\pi}[G_{t:t+2} | s_{t}], ..., \mathbb{E}_{\pi}[G_{t:t+n} | s_{t}] \right)$, such that $\mathbb{E}_{\pi}[G_{t:t+k} | s_{t}] = \sum_{i=p}^{q} v_{C_{i}, s_{t}}^{v_{s_{t}}(\tilde{N})}(\tilde{N})$, for $1 \leq k \leq n$, where $|C_{p-1}| < |C_{p}| = |C_{p+1}| = \cdots = |C_{q}| < |C_{q+1}|$ and $n < m$.
    \end{assumption}
    Recall that $v_{C, s}^{v_{s}(\tilde{N})} = v_{s}(\tilde{N}) v_{C, s}^{1}$, and $v_{C, s}^{1}$ in a given state $s$ is with identical definition as $v_{C}^{1}$. We can express $v_{C, s}^{v_{s}(\tilde{N})}$ as follows:
    \begin{equation*}
        v_{C, s}^{v_{s}(\tilde{N})}(D) = 
        \begin{cases}
            v_{s}(\tilde{N}) & \text{if } C \subseteq D, \\
            0 & \text{otherwise}.
        \end{cases}
    \end{equation*}
    By Definition~\ref{def:game-set}, as the coalition size $|C|$ increases, it requires experiencing more agent compositions to get the $v_{C, s}^{v_{s}(\tilde{N})}(\tilde{N})$ activated and thus $k_{C} v_{C, s}^{v_{s}(\tilde{N})}(\tilde{N})$ can be evaluated. In the context of NAHT, it implies that as the coalition size $|C|$ increases, the more agent interaction samples are needed to get an accurate estimation. This property is aligned with that using n-step returns $G_{t:t+n}$ to estimate expected value would require more samples (higher variance) during TD prediction, as $n$ increases~\cite{sutton1998reinforcement}. For this reason, if considering a simple case such that $|C_{1}| < |C_{2}|$, we may have that $v_{C_{1}, s_{t}}^{v_{s_{t}}(\tilde{N})}(\tilde{N}) = \mathbb{E}_{\pi}[G_{t:t+1} | s_{t}]$ and $v_{C_{2}, s_{t}}^{v_{s_{t}}(\tilde{N})}(\tilde{N}) = \mathbb{E}_{\pi}[G_{t:t+2} | s_{t}]$.

    Now, we consider sorting all possible coalitions, the subsets of $\tilde{N}$, according to their sizes, in a sequence. In this sequence, we can always have coalitions with identical size, such that $|C_{p-1}| < |C_{p}| = |C_{p+1}| = \cdots = |C_{q}| < |C_{q+1}|$, where $1 < p < q < 2^{|\tilde{N}|} - 1$. Following the above statement about the relation between n-step returns and $v_{C, s}^{v_{s}(\tilde{N})}(\tilde{N})$, it is viable to sum the $v_{C, s}^{v_{s}(\tilde{N})}(\tilde{N})$ of coalitions with the same size. By modifying $v_{C_{i}, s}^{v_{s}(\tilde{N})}(\tilde{N}) = \mathbb{E}_{\pi}[G_{t:t+i}]$ above, we relate the summand to a n-step return, such that $\mathbb{E}_{\pi}[G_{t:t+k} | s_{t}] = \sum_{i=p}^{q} v_{C_{i}, s_{t}}^{v_{s_{t}}(\tilde{N})}(\tilde{N})$, where $k < n$. In this way, we can match $v_{C, s_{t}}^{v_{s_{t}}(\tilde{N})}(\tilde{N})$ of all possible coalitions $C \subseteq \tilde{N}$ with a sequence of expected n-step returns $\left( \mathbb{E}_{\pi}[G_{t:t+1} | s_{t}], \mathbb{E}_{\pi}[G_{t:t+2} | s_{t}], ..., \mathbb{E}_{\pi}[G_{t:t+n} | s_{t}] \right)$.
    
\section{Complete Mathematical Proofs}
\label{sec:complete-math-proofs}
    \begingroup
    \def\thelemma{\ref{lemm:subgame-transform}}
        \begin{lemma}
            Given a fixed state $s \in {\cal S}$ each subgame $v_{s} \in G(s)$, it can be uniquely represented by $v_{s} = \sum_{\emptyset \neq C \subseteq \tilde{N}} \frac{k_{C}}{v_{s}(\tilde{N})} \cdot v_{C,s}^{v_{s}(\tilde{N})}$.
        \end{lemma}
        \begin{proof}
            We can represent each $k_{C}$ equivalently as $\frac{k_{C}v_{s}(\tilde{N})}{v_{s}(\tilde{N})}$. Substituting this term into the formula in Definition~\ref{def:game-set}, we can get the following formula:
            \begin{equation*}
                v_{s} = \sum_{\emptyset \neq C \subseteq \tilde{N}} \frac{k_{C}v_{s}(\tilde{N})}{v_{s}(\tilde{N})} \cdot v_{C, s}^{1}.
            \end{equation*}
            Since $v_{C, s}^{v_{s}(\tilde{N})} = v_{s}(\tilde{N}) v_{C, s}^{1}$ with $z = v_{s}(\tilde{N})$, we can have the following formula:
            \begin{equation*}
                v_{s} = \sum_{\emptyset \neq C \subseteq \tilde{N}} \frac{k_{C}}{v_{s}(\tilde{N})} \cdot v_{C, s}^{v_{s}(\tilde{N})}.
            \end{equation*}
        \end{proof}
    \endgroup
        
    \begingroup
    \def\theproposition{\ref{prop:shapley-machine-reward}}
        \begin{proposition}
            Given the condition that $\sum_{i=1}^{|\tilde{N}|} \phi_{i}(R_{t}) = R_{t}$, the payoff allocation defined on returns $R_{t}$, can be expressed as:
            \begin{equation*}
                \phi_{i}(R_{t}) := R_{t} - \sum_{j \neq i} \left( \ \phi_{j}(V(\tilde{N}, s_{t})) - \gamma \phi_{j}(V(\tilde{N}, s_{t+1})) \ \right).
            \end{equation*}
        \end{proposition}
        \begin{proof}
            Recall that $\phi(\cdot) \in \mathbb{R}^{|\tilde{N}|}$ is a multidimensional linear transformation, where $n$ is the number of agents. We now define $\phi(R_{t})$ by introducing Efficiency axiom. To satisfy the Efficiency axiom such that $\sum_{i=1}^{|\tilde{N}|} \phi_{i}(V(\tilde{N}, s_{t})) = V(\tilde{N}, s_{t})$, it is reasonable to assume that $\sum_{i=1}^{|\tilde{N}|} \phi_{i}(R_{t}) = R_{t}$. In other words, each agent's value expansion can be expressed independently with its own $\phi_{i}(R_{t})$, which can be justified by Theorem~\ref{thm:factored-dec-pomdp}. Next, we show how each $\phi_{i}(R_{t})$ is approximated by $\phi(V(\tilde{N}, s_{t}))$. It is not difficult to observe that for each agent $i \in \tilde{N}$, we have that $\phi_{i}(V(\tilde{N}, s_{t})) = \phi_{i}(R_{t}) + \gamma \phi_{i}(V(\tilde{N}, s_{t+1}))$. By the condition that $\sum_{i=1}^{|\tilde{N}|} \phi_{i}(R_{t}) = R_{t}$, we can derive that $\phi_{i}(R_{t}) = R_{t} - \sum_{j \neq i} \phi_{j}(R_{t})$. Since $\phi_{j}(R_{t}) = \phi_{j}(V(\tilde{N}, s_{t})) - \gamma \phi_{j}(V(\tilde{N}, s_{t+1}))$, we get the formula of $\phi_{i}(R_{t})$ that
            \begin{equation*}
                \phi_{i}(R_{t}) := R_{t} - \sum_{j \neq i} \left( \ \phi_{j}(V(\tilde{N}, s_{t})) - \gamma \phi_{j}(V(\tilde{N}, s_{t+1})) \ \right).
            \end{equation*}
        \end{proof}
    \endgroup

    \begin{proposition}
        For the class of superadditive games formed by a set of basis games $\{v^1_{C,s} | \emptyset \neq C \subseteq \tilde{N} \}$, it holds that $k_{C} \geq 0$, for all $\emptyset \neq C \subseteq \tilde{N}$.
    \end{proposition}
    \begin{proof}
        We first consider a characteristic function game space $G$ (a broader space of superadditive games). Each game $v: G \rightarrow \mathbb{R}_{\geq 0}$ belonging to $G$ can be uniquely represented as $$v = \sum_{\emptyset \neq C \subseteq \tilde{N}} k_{C} v_{C}^{1}.$$ The analytic form of $k_{C}$ under $G$ is represented as follows:
        \begin{equation*}
            k_{C} = \sum_{T \subseteq C} (-1)^{|C| - |T|} v(T),
        \end{equation*}
        where $v: G \rightarrow \mathbb{R}^{+}$ is a value function of a game belonging to $G$. The condition of superadditive games is as follows:
        \begin{equation*}
            v(T \cup C) \geq v(T) + v(C).
        \end{equation*}

        It has been proved that if the above condition holds, $k_{C} \geq 0$ has to be fulfilling~\cite{shapley1953value,dubey1975uniqueness}.
        
        We now extend the $G$ to a subgame space $G(s)$, for a fixed state $s \in \mathcal{S}$. The characteristics fulfilling in $G$ also holds in $G(s)$, since the subgame space $G(s)$ is isomorphic to a game space $G$. More detailed, every result above holds for $v_{s}: G(s) \rightarrow \mathbb{R}_{\geq 0}$, uniquely represented as $v_{s} = \sum_{\emptyset \neq C \subseteq \tilde{N}} k_{C} v_{C,s}^{1}$. 
        
        As we consider $G(s)$ as a space of superadditive games, $k_{C} \geq 0$ still holds.
    \end{proof}
    
    \begingroup
    \def\theproposition{\ref{prop:shapley-machine}}
        \begin{proposition}
            Shapley Machine is an algorithm that fulfills Efficiency, Additivity and Symmetry, so it learns $V_{i}$ as Shapley values for dynamic scenarios.
        \end{proposition}
        \begin{proof}
            Shapley Machine follows Efficiency, Additivity and Symmetry in implementation, as well as each subspace $G(s)$ is isomorphic to the original cooperative game space $G$. As per Theorem~\ref{thm:shapley-axioms}, we can get the conclusion.
        \end{proof}
    \endgroup
    
\section{Experimental Details}
\label{sec:experimental-details}
    \subsection{Implementation Details}
    \label{subsec:implementation-details}
        Our algorithm is built upon POAM and all loss functions for training the encoder-decoder model and related loss functions of PPO have been remained. Please refer to Appendix~\ref{subsec:poam} for details. For conciseness, we only list the novel loss functions proposed in this paper as below.
        
        \subsubsection{Shapley Machine}
            We now describe the details about the implementation of our proposed algorithm, referred to as Shapley Machine. In general, our algorithm is established based on POAM~\cite{wang2025n}, with modification to fulfill all the three axioms of Shapley value: Efficiency, Additivity and Symmetry. Since Symmetry has been implemented by incorporating embedding vectors to represent other agents' characteristics from each agent's view, we only need to fulfill Efficiency and Additivity. 
            
            \textbf{Implementing Additivity.} In general, POAM has implemented a version of Additivity using TD($\lambda$), which does not strictly conform to the principle of determining the number of n-step return components. Instead, we change TD($\lambda$) to TTD($\lambda$), for which the number of components is equal to the number of the non-empty coalitions in theory. Note that in some scenarios the episode length is smaller than the number of non-empty coalitions. In these cases, TTD($\lambda$) can be optionally reduced to the TD($\lambda$) for the finite-horizon tasks with the episode length as $T$, equivalently, the TTD($\lambda$) with $m=T$. Alternatively, we can select an empirical value of $m$ for each task. In this paper, we set $m=T$ for the scenarios 8v9 and 10v11.
        
            \textbf{Implementing Efficiency.} Recall that we have derived the condition for realizing the Efficiency axiom such that:
            \begin{equation*}
                R_{t,i} = R_{t} - \sum_{j \neq i} \left( V_{j}(s_{t}) - \gamma V_{j}(s_{t+1}) \right),
            \end{equation*}
            where $V_{j}(s_{t})$ indicates an agent $j$'s value estimation. During the practical training procedure the $V_{j}(s_{t})$ could be severely inaccurate in the beginning, which may result in the instability of learning. To mitigate this issue, we add an extra coefficient $\alpha \in (0, 1)$ to the term $\sum_{j \neq i} \left( V_{j}(s_{t}) - \gamma V_{j}(s_{t+1}) \right)$, such that:
            \begin{equation}
                R_{t,i} = R_{t} - \alpha \sum_{j \neq i} \left( V_{j}^{\theta^{c}}(s_{t}) - \gamma V_{j}^{\theta^{c}}(s_{t+1}) \right).
            \label{eq:shaping-reward-implementation}
            \end{equation}
            This coefficient $\alpha$ can be either manually set up as a fixed value, or implemented by a scheduler starting from $0$ to some preset upper limit. $R_{t}$ in $L_{\theta}(h_{t,i}, a_{t,i})$ and $L_{\theta^{c}}(h_{t,i})$ is replaced by the above $R_{t,i}$. To clarify this change, we the two new losses are expressed as: $\hat{L}_{\theta}(h_{t,i}, a_{t,i})$ and $\hat{L}_{\theta^{c}}(h_{t,i})$.
        
            Furthermore, it is needed to search the underlying factorization scheme of the transition function in Dec-POMDP, according to Theorem~\ref{thm:factored-dec-pomdp}. To fulfill this, we need to fulfill the following condition:
            \begin{equation}
            \label{eq:efficiency-eq}
                \sum_{i=1}^{|\tilde{N}|} V_{i}(s_{t}) = V(\tilde{N}, s_{t}).
            \end{equation}
            The above equality is implemented as a regularization term during training. In the main paper, we have defined a Bellman equation characterizing the global value $V(\tilde{N}, s_{t})$, such that
            \begin{equation*}
                V(\tilde{N}, s_{t}) = \mathbb{E}_{\pi} \left[ R_{t} + \gamma V(\tilde{N}, s_{t+1}) \mathlarger{\mathlarger{|}} s_{t} \right].
            \end{equation*}
            To maintain consistency with the critic losses, we consider $\lambda$-returns and the above equation can be extended as:
            \begin{equation*}
                V(\tilde{N}, s_{t}) = \mathbb{E}_{\pi} \left[ G_{t}^{\lambda} | s_{t} \right],
            \end{equation*}
            where each $G_{t:t+n}$ forming the $G_{t}^{\lambda}$ is expressed as follows:
            \begin{equation*}
                G_{t:t+n} = R_{t} + \gamma R_{t+1} + \cdots + \gamma^{n-1} R_{t+n} + \gamma^{n} V(\tilde{N}, s_{t+n}).
            \end{equation*}
            By incorporating the condition Eq.~\eqref{eq:efficiency-eq} into the above formula, we can obtain the regularization term referred to as the \textbf{efficiency loss}, as follows:
            \begin{equation*}
                L_{\theta^{c}}^{e}(h_{t,i}) = \frac{1}{2} \left( \hat{G}_{t}^{\lambda} - \sum_{i=1}^{|\tilde{N}|} V_{i}^{\theta^{c}}(s_{t}) \right)^{2},
            \end{equation*}
            where each $\hat{G}_{t:t+n}$ forming the $\hat{G}_{t}^{\lambda}$ is expressed as follows:
            \begin{equation*}
                \hat{G}_{t:t+n} = R_{t} + \gamma R_{t+1} + \cdots + \gamma^{n} R_{t+n} + \gamma^{n} \sum_{i=1}^{|\tilde{N}|} V_{i}^{\theta^{c}}(s_{t+n}).
            \end{equation*}
        
            In summary, the total loss function of Shapley Machine is as follows:
            \begin{equation*}
                L_{\texttt{SM}} = \frac{1}{T} \sum_{t=1}^{T} \left( \sum_{i \in \tilde{C}_{t}} \hat{L}_{\theta}(h_{t,i}, a_{t,i}) + \beta_{1} \sum_{i \in \tilde{N}} \hat{L}_{\theta^{c}}(h_{t,i}) + \beta_{2} L_{\theta^{c}}^{e}(h_{t,i}) \right),
            \end{equation*}
            where $T$ is the episode length; $\beta_{1}, \beta_{2} \in (0, 1)$ are two coefficients to control the importance of the two losses; as well as $\tilde{C}_{t}$ indicates the controlled agent set at timestep $t$ and $\tilde{N}$ indicates the ad hoc team following the convention in \cite{wang2025n}.

        \subsubsection{POAM}
            The implementation of POAM follows \cite{wang2025n}, which has been detailed in Appendix~\ref{subsec:poam}. In summary, the total loss function of POAM is as follows:
            \begin{equation*}
                L_{\texttt{POAM}} = \frac{1}{T} \sum_{t=1}^{T} \left( \sum_{i \in \tilde{C}_{t}} L_{\theta}(h_{t,i}, a_{t,i}) + \beta_{1} \sum_{i \in \tilde{N}} L_{\theta^{c}}(h_{t,i}) \right).
            \end{equation*}

        \subsection{Banzhaf Machine}
            The total loss function and the implementation of Banzhaf Machine are close to POAM, the only difference for which is that Banzhaf Machine uses TTD($\lambda$) instead of TD($\lambda$) that has been used in POAM.

    \subsection{Experimental Domains}
    \label{subsec:experimental-domains}
        We now briefly introduce the experimental domains for running experiments. If one would like to know more about details, please refer to \cite{wang2025n}.
        \subsubsection{MPE-PP} 
            The mpe-pp environment is a predator-prey task implemented within the Multi-Agent Particle Environment (MPE) framework. It simulates interactions within a two-dimensional space populated by two static obstacles, where three pursuer agents must cooperate to capture a single adversarial evader. A successful capture is defined as at least two pursuers simultaneously colliding with the evader, upon which the pursuers receive a positive reward of +1. If the capture is unsuccessful, no reward is granted. This environment is designed to test the ability of agents to coordinate under spatial and dynamic constraints.

        \subsubsection{SMAC}
            The StarCraft Multi-Agent Challenge (SMAC) serves as a benchmark suite for evaluating MARL algorithms in partially observable, cooperative settings. Built atop the StarCraft II game engine, SMAC presents a variety of micromanagement tasks where each agent (e.g., a Marine or Stalker) operates based on limited local observations and must coordinate actions with teammates to overcome enemy units. In this work, we focus on four specific SMAC scenarios:
            \textbf{5v6}: Five allied Marines versus six enemy Marines,
            \textbf{8v9}: Eight allied Marines versus nine enemy Marines,
            \textbf{10v11}: Ten allied Marines versus eleven enemy Marines,
            \textbf{3s5z}: Three allied Stalkers versus five enemy Zealots.
            At each timestep, agents receive a shaped reward proportional to the damage they inflict on opponents, along with bonus rewards of 10 points for each enemy defeated and 200 points for achieving overall victory by eliminating all adversaries. The total return is normalized so that the maximum achievable return in each scenario is 20. The action space in SMAC is discrete, enabling each agent to choose actions such as attacking a particular enemy, moving in a specific direction, or remaining idle. Notably, the variation in the number and type of agents and opponents across tasks results in scenario-specific observation and action space dimensionalities, thereby introducing further diversity and complexity for algorithmic evaluation. The length of an episode varies across different scenarios: MPE with $T=100$, 3sv5z with $T=250$, 5v6 with $T=70$, 8v9 with $T=120$ and 10v11 with $T=150$.
            
    \subsection{Experimental Settings}
    \label{subsec:experimental-settings}
        The training procedure of the n-agent ad hoc teamwork (NAHT) process is briefly introduced here. For more details, please refer to \cite{wang2025n}. For each scenario (e.g. MPE, 3sv5z, 5v6, 8v9 and 10v11), there are five groups of pretrained agents acting the uncontrolled agents such as IQL, MAPPO, VDN, QMIX and IPPO. For each episode evaluation, the number of uncontrolled agents $\tilde{u}$ is sampled, and then a group of the pretrained agents is sampled. Given that each task specifies a fixed total number of agents $\tilde{n}$, the number of controlled agents is computed as $\tilde{n} - \tilde{u}$. \textbf{Note that this is still a special case of openness, since for a varying number of controlled agents the number of uncontrolled agents is also varied in correspondence, though with an upper limit due to the task specifications.} The distributions for all sampling procedures are modeled as the multinational distribution with no replacement. Each uncontrolled agent conducts the greedy policy in both the training procedure and the testing procedure. In contrast, each controlled agent conducts the on-policy sampling by the parameterized policy (PPO in our algorithm) during the training procedure, while the greedy policy during the testing procedure.

    \subsection{Hyperparameter Settings}
    \label{subsec:hyperparameter-settings}
        Since our algorithm is established based on POAM, most hyperparameter settings follow that in \cite{wang2025n}. First, the actors and critics are implemented in recurrent neural networks, with full parameter sharing. Specifically, they are implemented by two fully connected layers followed by a GRU layer and an output layer. Each layer has 64 dimensions with a ReLU activation function, and layer normalization is applied. The encoder-decoder networks for inferring agent characteristics are also implemented in parameter sharing. The encoder is implemented by a GRU layer, followed by a fully connected layer with a ReLU activation function and an output layer. The decoder is implemented by two fully connected layers with ReLU activation functions, followed by an output layer. Adam Optimizer is used to train all models. The detailed hyperparameter for experiments is shown in Tables~\ref{tab:shapley-machine-hyperparams} and \ref{tab:common-hyperparams}. For the scenarios such as MPE, 3sv5z and 5v6, the values of $m$ are set as the number of non-empty coalitions. For the scenarios such as 8v9 and 10v11, the values of $m$ is simply set as the length of an episode (see Appendix~\ref{subsec:implementation-details} for more details). Note that the term $\sum_{j \neq i} \left( V_{j}^{\theta^{c}}(s_{t}) - \gamma V_{j}^{\theta^{c}}(s_{t+1}) \right)$ in shaped rewards has been standardised to match the scales of standardised rewards. To stabilize the learning processes for large-scale scenarios, the value loss clip have been added to Shapley Machine for 8v9 and 10v11.
        
        \begin{table}[ht]
        \caption{Key hyperparameter of Shapley Machine: $\lambda$ indicates the parameter for the weighting functions of TD($\lambda$), $m$ indicates the number of basis games, $\alpha$ controls the importance of the term $\sum_{j \neq i} \left( V_{j}^{\theta^{c}}(s_{t}) - \gamma V_{j}^{\theta^{c}}(s_{t+1}) \right)$ in Eq.~\eqref{eq:shaping-reward-implementation}, $\beta_{1}$ controls the importance of the critic loss $\sum_{i \in \tilde{N}} L_{\theta^{c}}(h_{t,i})$, and $\beta_{2}$ controls the importance of the efficiency loss $L_{\theta^{c}}^{e}(h_{t,i})$.}
        \label{tab:shapley-machine-hyperparams}
        \centering
        \begin{tabular}{cccccc}
            \hline\hline
            Task & $\lambda$ & $m$ & $\alpha$ & $\beta_{1}$ & $\beta_{2}$ \\
            \hline\hline
            MPE & $0.85$ & $7$ & $0.01$ & $0.5$ & $0.01$ \\
            3sv5z & $0.85$ & $7$ & $0.01$ & $0.5$ & $0.01$ \\
            5v6 & $0.85$ & $31$ & $0.002$ & $0.5$ & $0.001$ \\
            8v9 & $0.95$ & $120$ & $0.01$ & $0.5$ & $0.001$ \\
            10v11 & $0.95$ & $150$ & $0.01$ & $0.5$ & $0.001$ \\
            \hline\hline
        \end{tabular}
        \end{table}

        \begin{table}[ht]
        \caption{Common hyperparameter for RL settings.}
        \label{tab:common-hyperparams}
        \centering
        \begin{tabular}{cc}
            \hline\hline
            Hyperparameter & Value \\
            \hline\hline
            LR & $0.0005$ \\
            Epochs & $5$ \\
            Minibatches & $1$ \\
            Buffer size & $256$ \\
            Entropy coefficient& $0.05$ \\
            Clip & $0.2$ \\
            ED LR & $0.0005$ \\
            ED epochs & $1$ \\
            ED Minibatches & $1$ \\
            Optim\_alpha (Adam) & $0.99$ \\
            Optim\_eps (Adam) & $0.00001$ \\
            Use\_obs\_norm & True \\
            Use\_orthogonal\_init & True \\
            Use\_adv\_std & True \\
            Standardise\_rewards & True \\
            num\_parallel\_envs & 8 \\
            \hline\hline
        \end{tabular}
        \end{table}

    \subsection{Computational Resources}
        All experiments are conducted on Intel Xeon Gold 6230 CPUs and Nvidia V100-SXM2 GPUs. Each experiment run on MPE takes approximately 7 hours, utilizing 20 CPU cores and 1 GPU. Each experiment run on SMAC takes between 8 and 19 hours, utilizing 30 CPU cores and 1 GPU. All experiments are trained with 20M timesteps.
        
\section{Additional Experiments}
\label{sec:additional-experiments}
    \subsection{Analyzing Effects of Efficiency Axioms}
    \label{subsec:analyzing-effects-of-efficiency-axioms}
        We now demonstrate the association between the changing of the mean shaping rewards and the corresponding response of testing returns. As shown in Figures~\ref{fig:mpe-efficiency-demo}, the shaped rewards normally increases in the beginning and then gradually decreases to a plateau. The efficiency loss also exhibits the similar phenomenon. As shown in Figure~\ref{fig:5v6-efficiency-demo}, it can be observed that when the perturbation happens during training (as evidenced by test returns), the shaped rewards can consistently manifest this situation (as highlighted in red vertical dashed lines). This feature is crucial in NAHT, since it is a common case when an unseen agent appears in the environment, which may perturb the training stability. This justifies the practical effectiveness of the shaping rewards $R_{t,i}$ we propose to fulfill the Efficiency axiom.
        
        \begin{figure}[ht!]
            \centering
            \includegraphics[width=\linewidth]{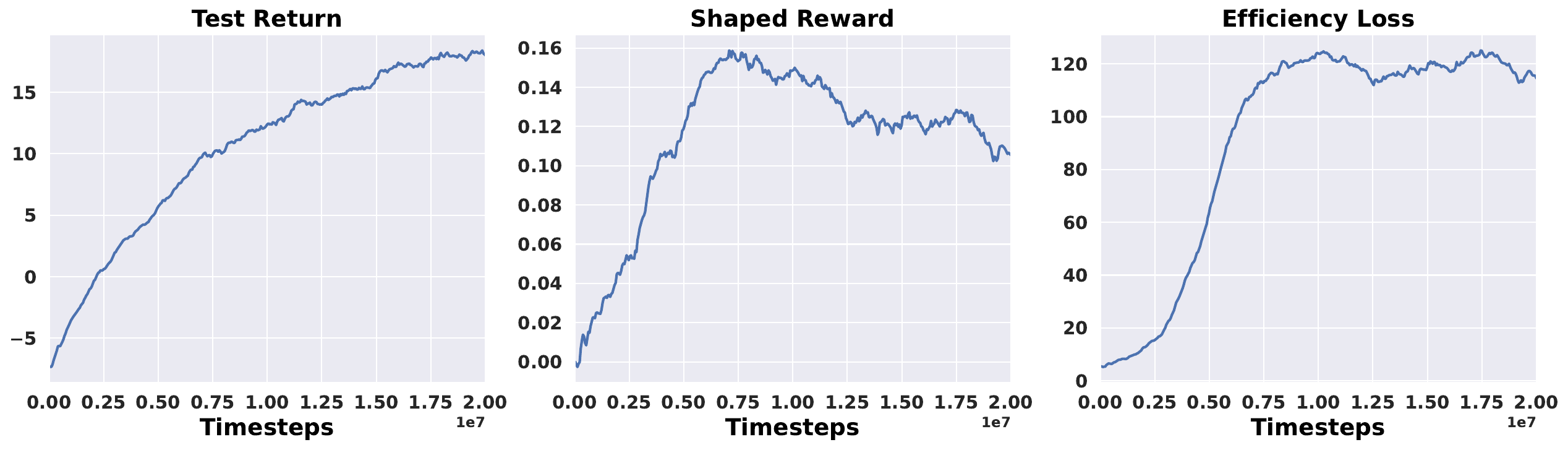}
            \caption{One run of MPE to show the variation of test return, shaped reward and efficiency loss.}
        \label{fig:mpe-efficiency-demo}
        \end{figure}
    
        \begin{figure}[ht!]
            \centering
            \includegraphics[width=\linewidth]{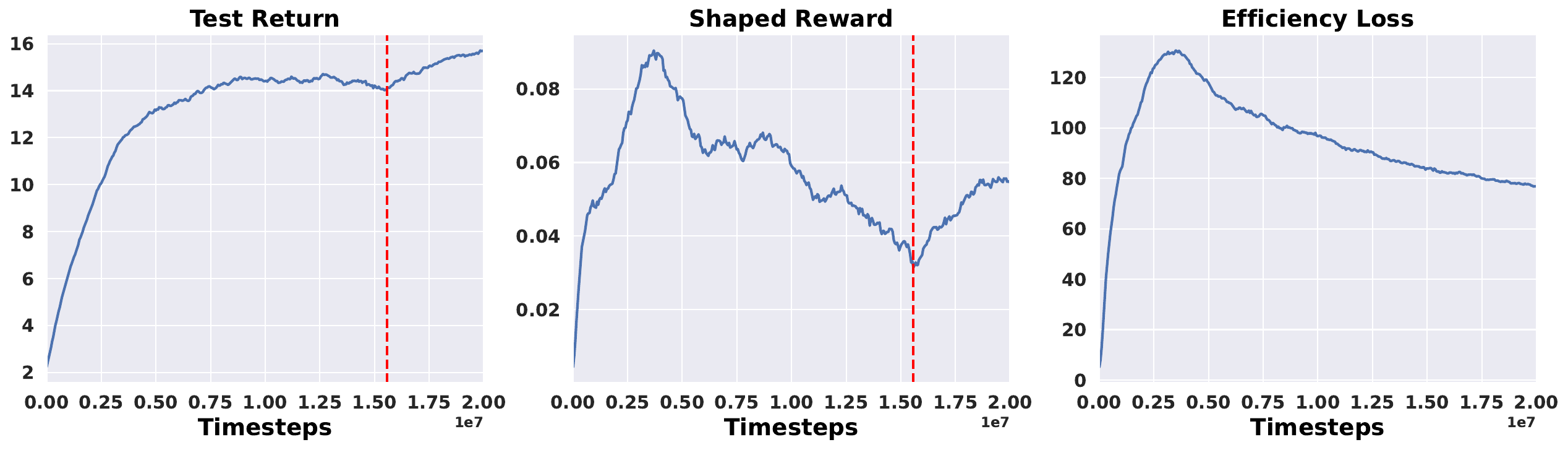}
            \caption{One run of 5v6 to show the variation of test return, shaped reward and efficiency loss.}
        \label{fig:5v6-efficiency-demo}
        \end{figure}

    \subsection{Potential Reasons for Instability for Large-Scale Scenarios}
    \label{subsec:potential-reasons-for-instability-for-large-scale-scenarios}
        During our trials, we observed that the performance of Shapley Machine was unstable in large-scale scenarios, such as 8v9 and 10v11. We now analyze the potential reason behind this phenomenon. As observed from Figures~\ref{fig:8v9-evidence-of-value-instability} and \ref{fig:10v11-evidence-of-value-instability}, it can be justified that the instability of shaped rewards is the key reason to cause the unstable learning progresses for large-scale scenarios. As discussed in Appendix~\ref{subsec:analyzing-effects-of-efficiency-axioms}, we hypothesize that the instability is caused by the the shaping reward terms constituted of instable changing of values (critics), such as $\sum_{j \neq i} \left( V_{j}^{\theta^{c}}(s_{t}) - \gamma V_{j}^{\theta^{c}}(s_{t+1}) \right)$. To mitigate this issue, we propose to add the value loss clip to stabilize the learning progresses of value functions for the scenarios such as 8v9 and 10v11.
        
        \begin{figure}
            \centering
            \includegraphics[width=0.65\linewidth]{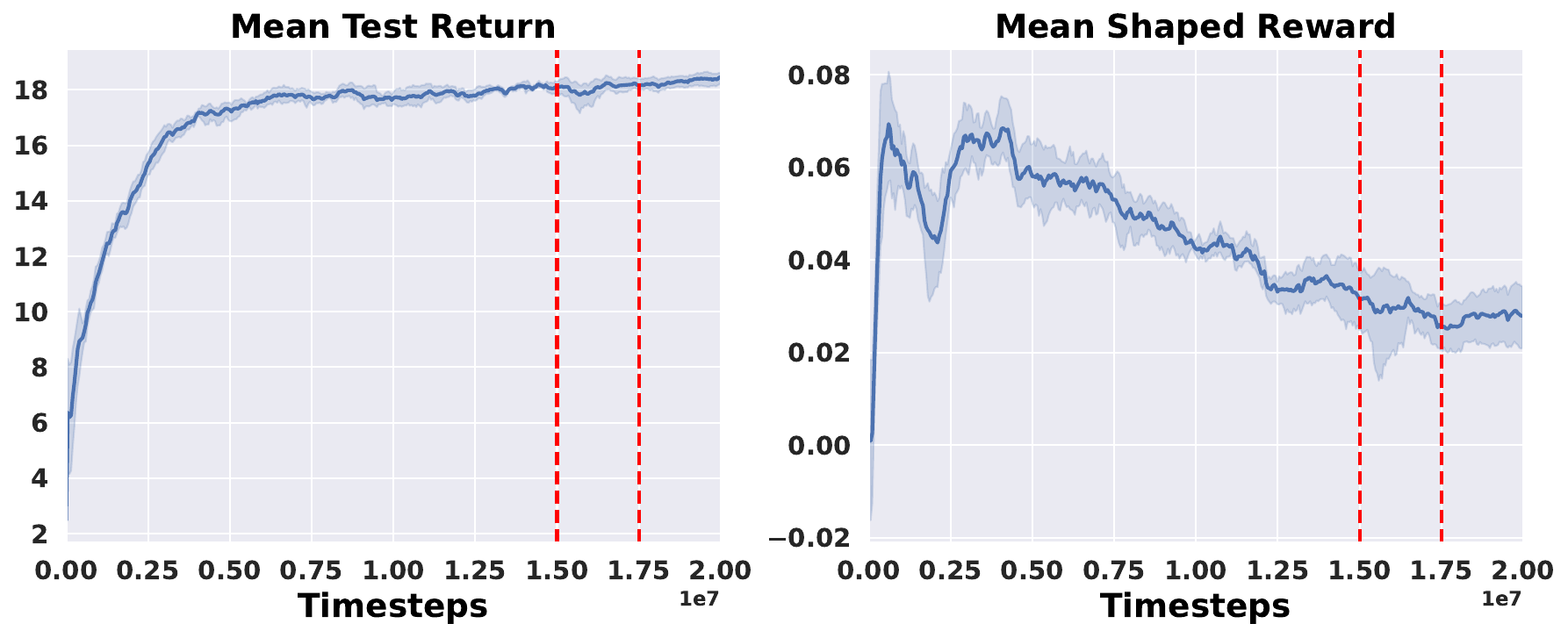}
            \caption{A case of 8v9 to show the instability of learning progress via the mean shaped rewards. The range bounded by two red dashed lines shows the fluctuation of the learning process.}
        \label{fig:8v9-evidence-of-value-instability}
        \end{figure}
        
        \begin{figure}
            \centering
            \includegraphics[width=0.65\linewidth]{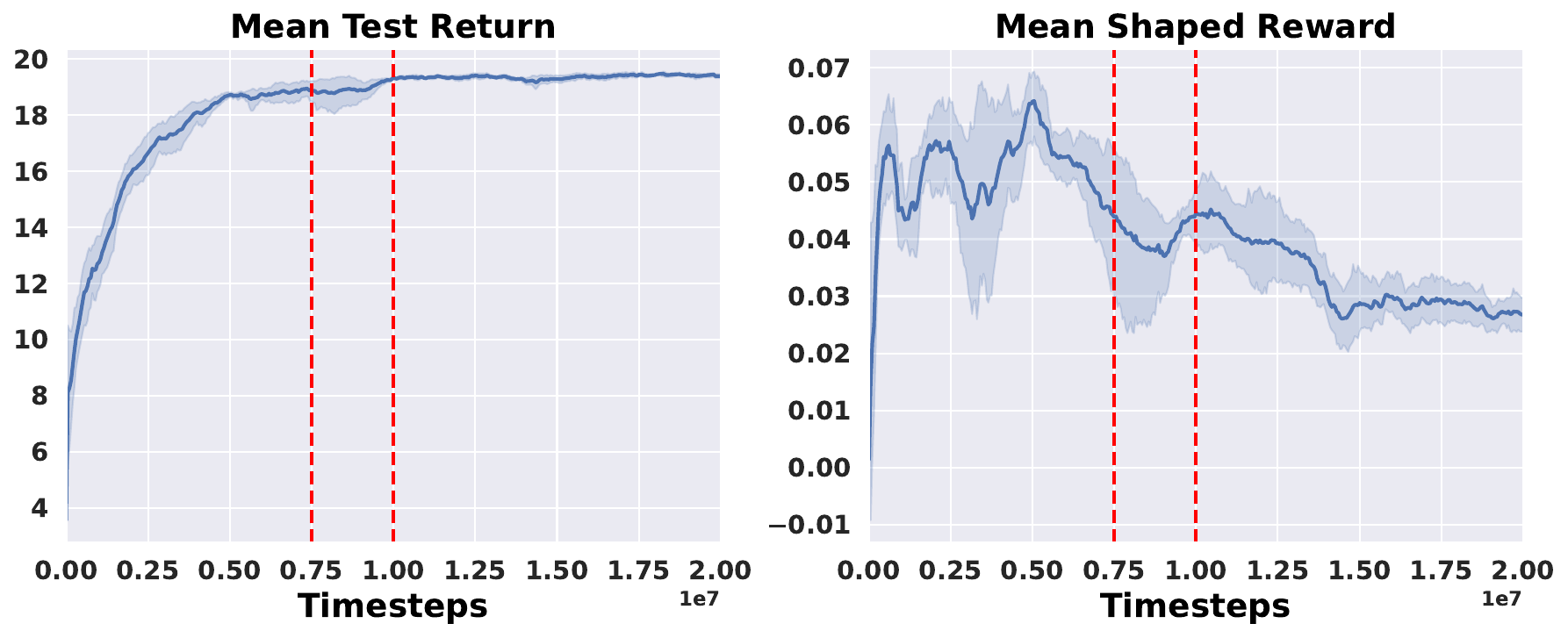}
            \caption{A case of 10v11 to show the instability of learning progress via the mean shaped rewards. The range bounded by two red dashed lines shows the fluctuation of the learning process.}
        \label{fig:10v11-evidence-of-value-instability}
        \end{figure}
        
        As shown in Figure~\ref{fig:result-value-loss-clip}, the learning processes with the value loss clip apparently perform more stable than those without. This verifies our initial hypothesis. Due to this strategy will slow down and even hinder the performance for the other three scenarios, we have not posted this as a common strategy. \textbf{We believe this warrants further investigation in future work before any claims can be made about its general performance.} 
        
        \begin{figure}
            \centering
            \includegraphics[width=0.65\linewidth]{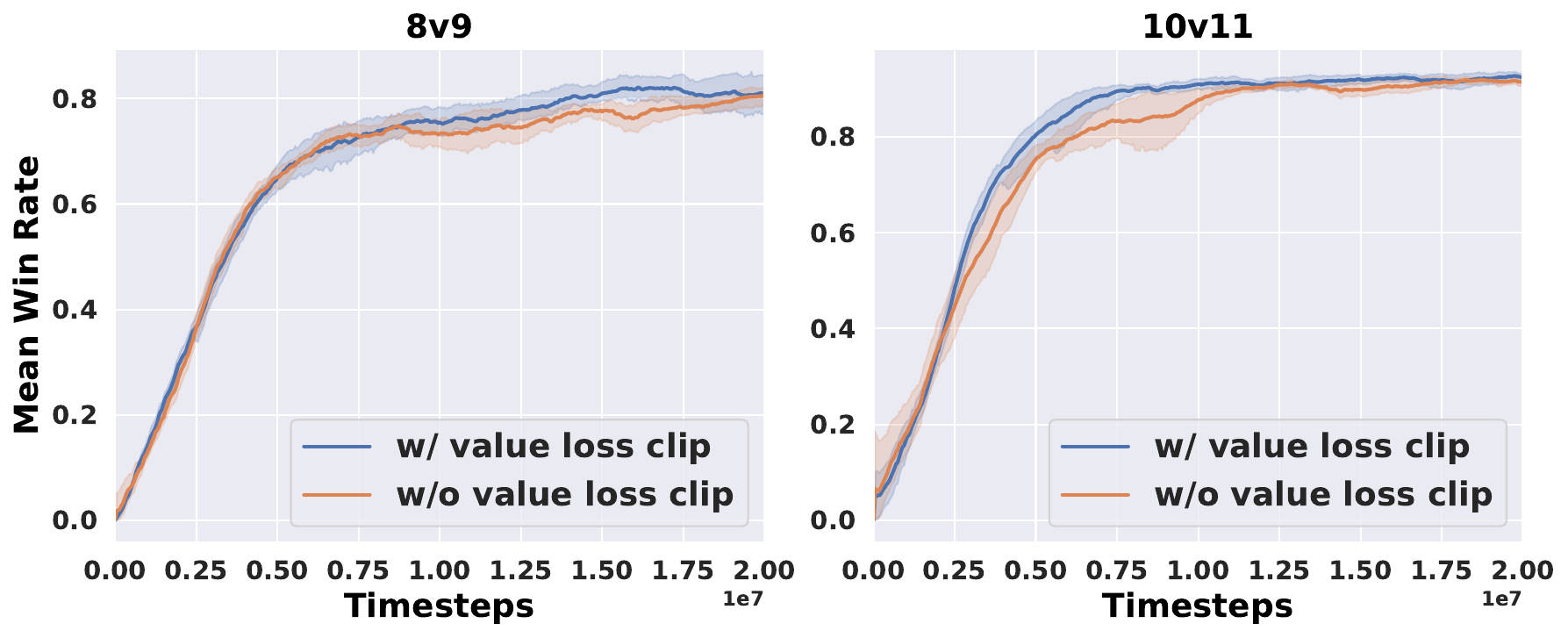}
            \caption{Comparison between stability of learning processes with and without the value loss clip in the 8v9 and 10v11 scenarios.}
        \label{fig:result-value-loss-clip}
        \end{figure}
        
    \subsection{Further Evidence on the Benefit of Structured Algorithm}
    \label{subsec:further-evidence-on-the-benefit-of-structured-algorithm}
        In addition to the presentation of critic losses in the main of paper, we also present the entropy of policies in Figure~\ref{fig:entropy}. Consistent with the critic losses, the entropy of Shapley Machine decreases faster than POAM, which is a strong evidence to show the power of structured value functions in Shapley values as each agent's credit during training.
        
        \begin{figure}[ht!]
            \centering
            \includegraphics[width=\linewidth]{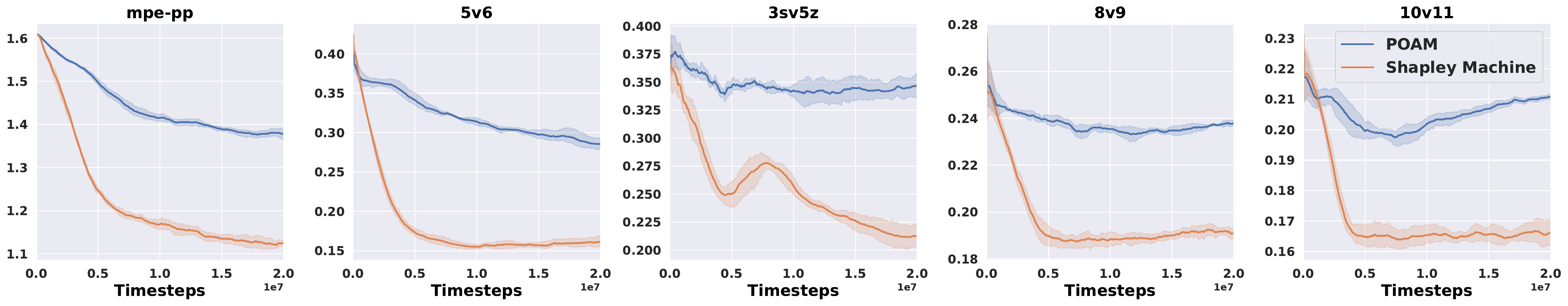}
            \caption{Entropy of policies across all scenarios during the training procedure.}
        \label{fig:entropy}
        \end{figure}
        
    \subsection{Relationship between the Number and the Values of Basis Games}
    \label{subsec:relationship-between-the-number-and-the-values-of-basis-games}
        The values of basis games are set up by the probability values of a geometric distribution governed by $\lambda$ in this paper, respecting the convention of RL. As a result, the $\lambda$ value will influence the shape of the geometric distribution. When the $\lambda$ value grows larger, the geometric distribution tends to be a longer tail with the probability decays slowly as $m$ increases, and vice versa. In other words, \textbf{controlling the number of basis games can be implemented by simply changing the values of basis games, or equivalently, the $\lambda$ value, under the assumption of the geometric distribution}. This is consistent with the insight from RL. To verify this result still holding in NAHT, we conduct a case study, as shown in Figure~\ref{fig:control-number-of-basis-equal-to-change-values}. It can be observed that the number of active basis games for $m = 31$ and $m = 70$ is the same under $\lambda = 0.85$, resulting in nearly identical learning performance. \textbf{This aligns with the conventional view of $\lambda$'s impact on TD($\lambda$) in reinforcement learning, primarily as a mechanism for controlling the tail length of the return.}
        
        \begin{figure}[htbp]
          \centering
          \begin{subfigure}[b]{0.26\textwidth}
            \centering
            \includegraphics[width=\linewidth]{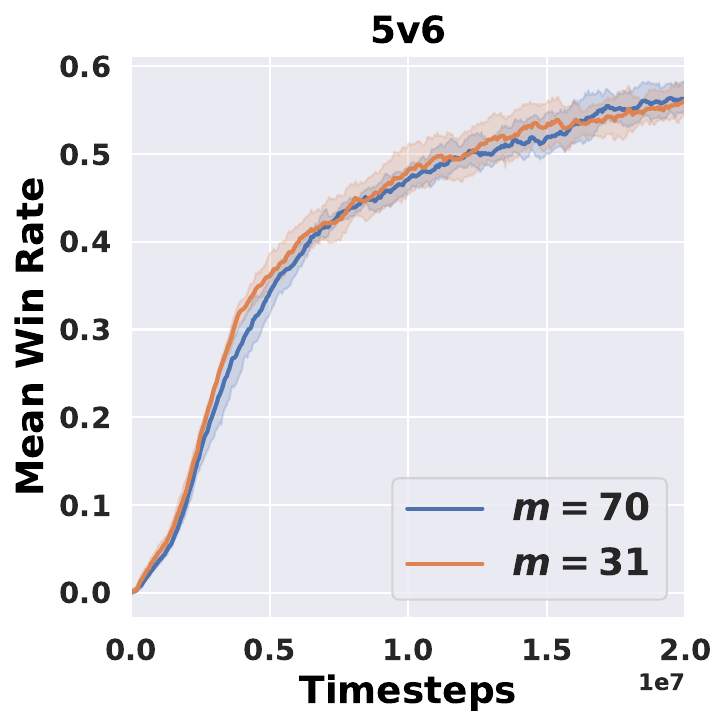}
            \caption{Shapley Machine: $m=31$ vs. $m=70$ for $\lambda=0.85$.}
          \end{subfigure}
          ~
          \begin{subfigure}[b]{0.64\textwidth}
            \centering
            \includegraphics[width=\linewidth]{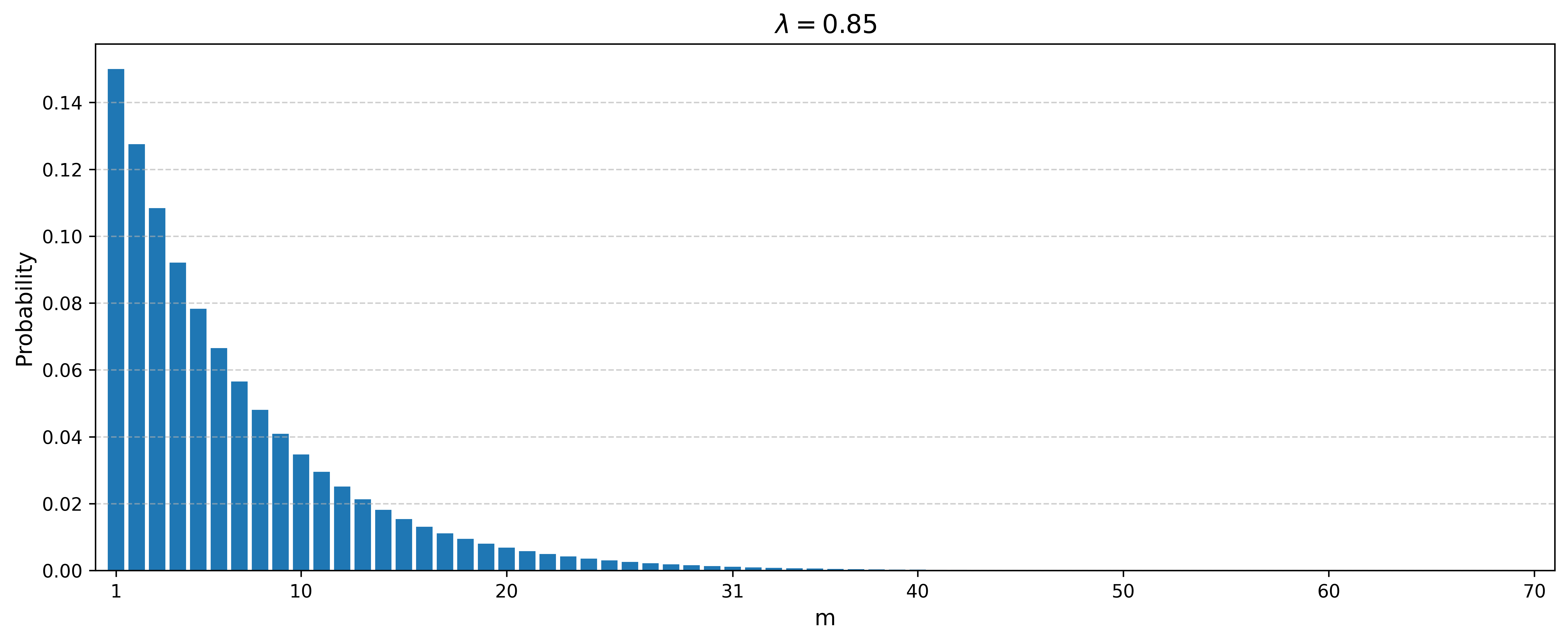}
            \caption{Demonstration of the shape the basis game values with $m=7$ and $\lambda=0.85$.}
          \end{subfigure}
        \caption{Controlling the number of active basis games can be implemented by changing the parameters to control the shape of basis game values. Note that the probability values for $m > 31$ are nearly zero.}
        \label{fig:control-number-of-basis-equal-to-change-values}
        \end{figure}
        
        In addition, we demonstrate that even with the same number of basis games, their specific values significantly influence the learning process. As shown in Figure~\ref{fig:different-values-same-number-basis-games}, given the same number of active basis games as $m=7$, the learning processes with $\lambda=0.5$ and $\lambda=0.85$ result in diverse performance. \textbf{This provides an alternative insight into how $\lambda$ influences TD($\lambda$), diverging from the traditional reinforcement learning perspective.}
        
        \begin{figure}[htbp]
          \centering
          \begin{subfigure}[b]{0.26\textwidth}
            \centering
            \includegraphics[width=\linewidth]{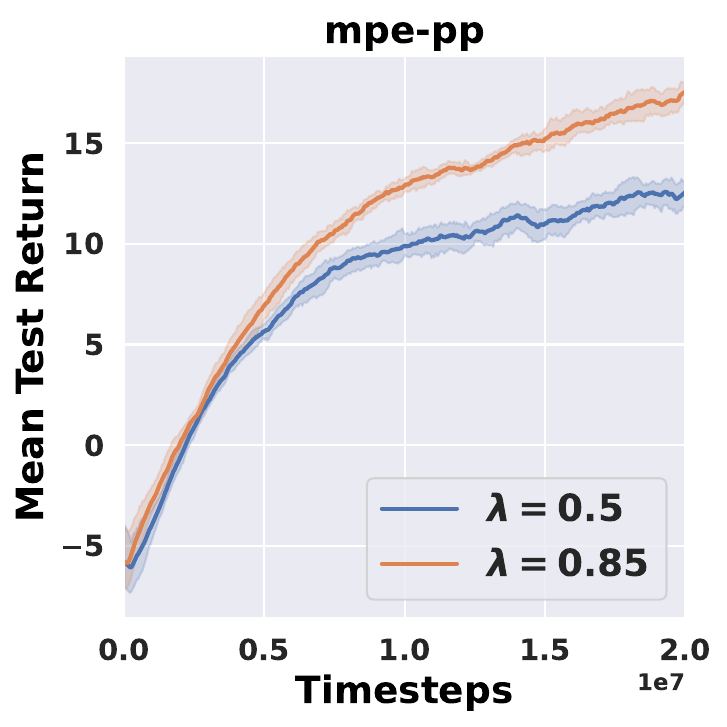}
            \caption{Shapley Machine: $\lambda=0.5$ vs. $\lambda=0.85$ for $m=7$.}
            \label{fig:shapley-machine-m-mpe}
          \end{subfigure}
          ~
          \begin{subfigure}[b]{0.64\textwidth}
            \centering
            \includegraphics[width=\linewidth]{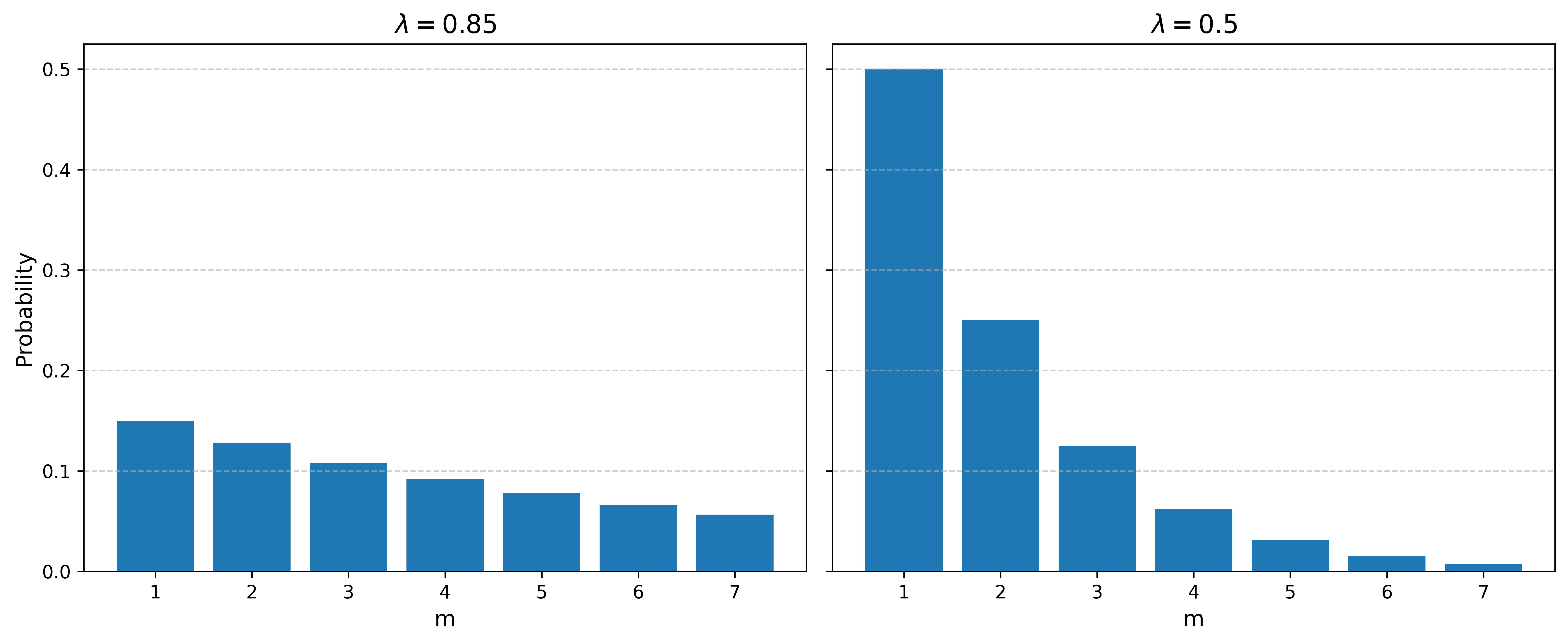}
            \caption{Demonstration of the shapes of the basis game values with $m=7$, for $\lambda=0.5$ and $\lambda=0.85$, respectively.}
            \label{fig:shapley-machine-m-demo-mpe}
          \end{subfigure}
        \caption{The influence of different values of basis games given the same number of active basis games. Note that the shapes of the geometric distributions vary significantly with different values of $\lambda$, under the same number of active $m$.}
        \label{fig:different-values-same-number-basis-games}
        \end{figure}
        
    \subsection{Empirical $m$ values and Number of Agents}
    \label{subsec:empirical-m-values-and-number-of-agents}
        We now investigate the relationship between the empirical $m$ values and the number of agents, to tackle some scenarios where the length of an episode $T$ is far smaller than the theoretical value of $m$, especially when the number of agents is large. We have demonstrated that for the case with $5$ agents the empirical effective value of $m$ is $20$. To prove the claim that the empirical $m$ values are still proportional to the number of agents, it is sufficient to show that $m=20$ cannot work as well as the $m$ set by larger values for the scenarios with more agents, such as the 8v9 and 10v11 scenarios. As shown in Figure~\ref{fig:evidence-of-num-agent-and-num-basis}, both the 8v9 and 10v11 scenarios demonstrate the better performance on the larger $m$ values, set by the length of an episode. Therefore, we can conclude a law of empirical $m$ values with respect to the number of agents, as shown in Figure~\ref{fig:empirial-m-law}. \textbf{Empirically, $m$ tends to grow as the number of agents increases.}
        
        \begin{figure}[htbp]
            \centering
            \includegraphics[width=0.65\linewidth]{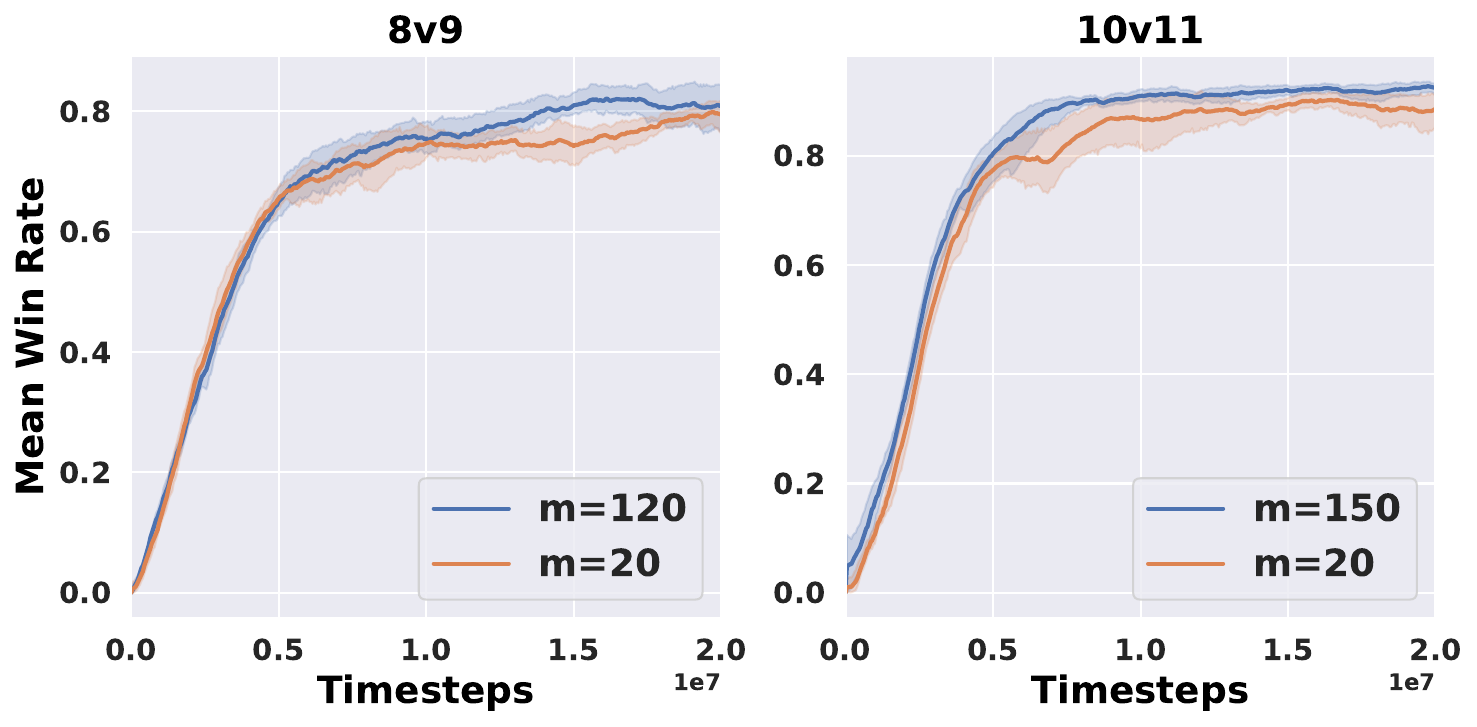}
            \caption{Comparison between the mean win rates of $m=20$ and $m=120$ in the 8v9 and 10v11 scenarios.}
        \label{fig:evidence-of-num-agent-and-num-basis}
        \end{figure}
        
        \begin{figure}[htbp]
            \centering
            \includegraphics[width=0.6\linewidth]{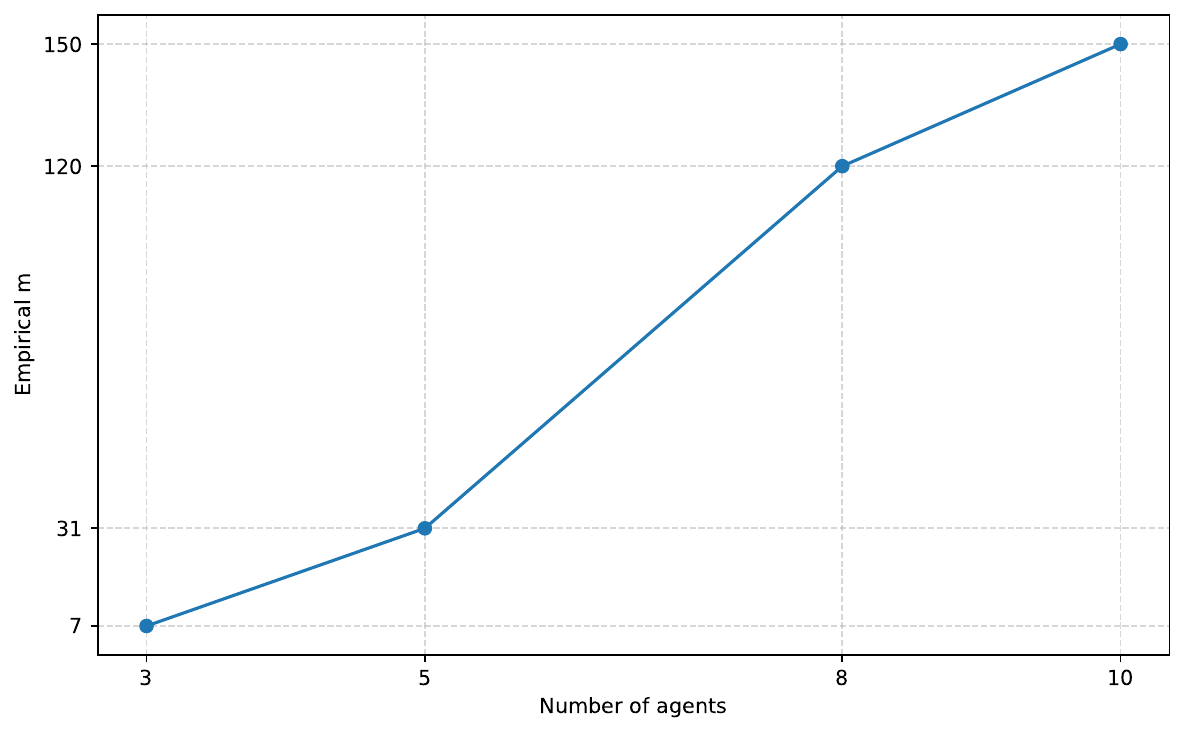}
            \caption{The law of empirical $m$ values with respect to the number of agents.}
        \label{fig:empirial-m-law}
        \end{figure}
        
    \subsection{Ablation Studies}
    \label{subsec:ablation-studies}
        We conducted ablation studies across all experimental domains to assess the effectiveness of our proposed algorithm, Shapley Machine. As shown in Figure~\ref{fig:ablation}, all the shaped reward terms, efficiency loss components and truncated TD($\lambda$) inspired by the cooperative games contribute meaningfully to performance in most scenarios.

        Specifically, in the MPE, 5v6 and 3sv5z scenarios, the superior performance of the Shapley Machine variant with $\alpha=0$ and $\beta_2=0$ over POAM highlights the effectiveness of using truncated horizons in TD($\lambda$). In contrast, in the 8v9 and 10v11 scenarios, the setting of POAM is equivalent to this Shapley Machine variant, except that a value loss clip is applied in the latter. Therefore, the lack of superior performance by the Shapley Machine in these cases suggests that the value loss clip is not the primary factor driving its performance gains. Instead, it only works as a strategy to stabilize the performance of Shapley Machine for large-scale scenarios, as shown in the previous experimental results. 
        
        Another noteworthy observation is that variants using only the efficiency loss still demonstrate competitive performance. This suggests that, in certain settings, the efficiency loss can approximate the Efficiency axiom, even if it does not formally guarantee that the resulting reward function satisfies it. Nonetheless, results from all scenarios except 8v9 highlight that fulfilling the Efficiency axiom on the reward function is critical for achieving effective performance. This indicates that earlier Shapley value-based MARL methods~\cite{wang2020shapley,wang2022shaq}, which did not enforce the Efficiency axiom on the reward, may have been less effective in scenarios where this axiom plays a pivotal role.
        
        \begin{figure}
            \centering
            \includegraphics[width=0.19\linewidth]{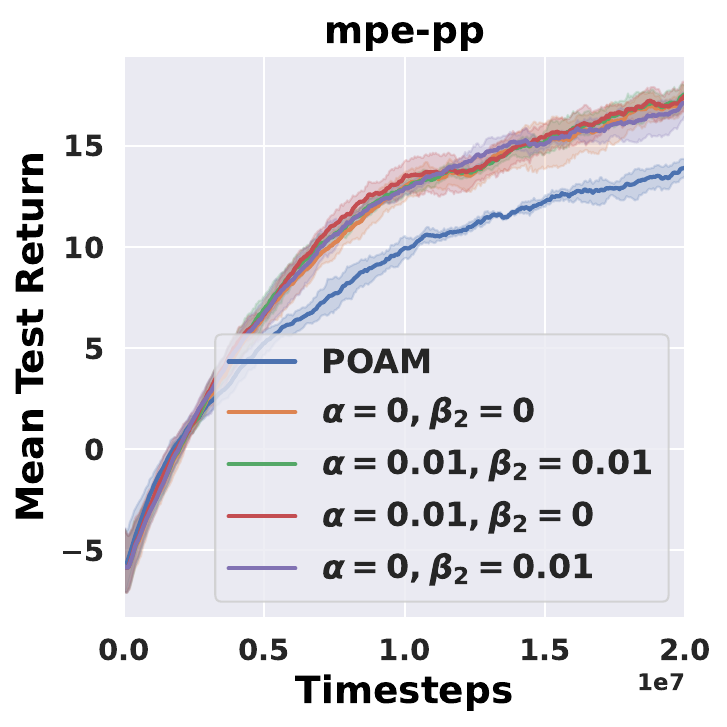}
            \includegraphics[width=0.19\linewidth]{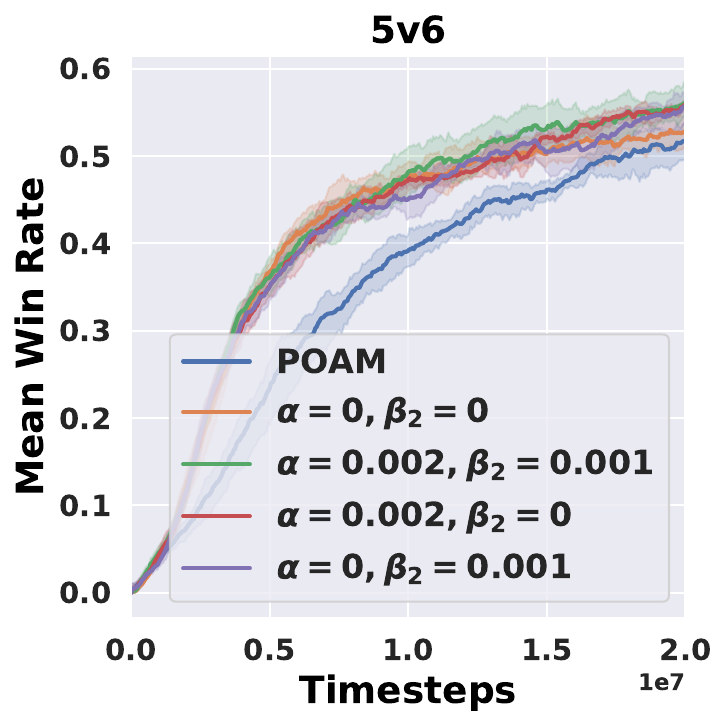}
            \includegraphics[width=0.19\linewidth]{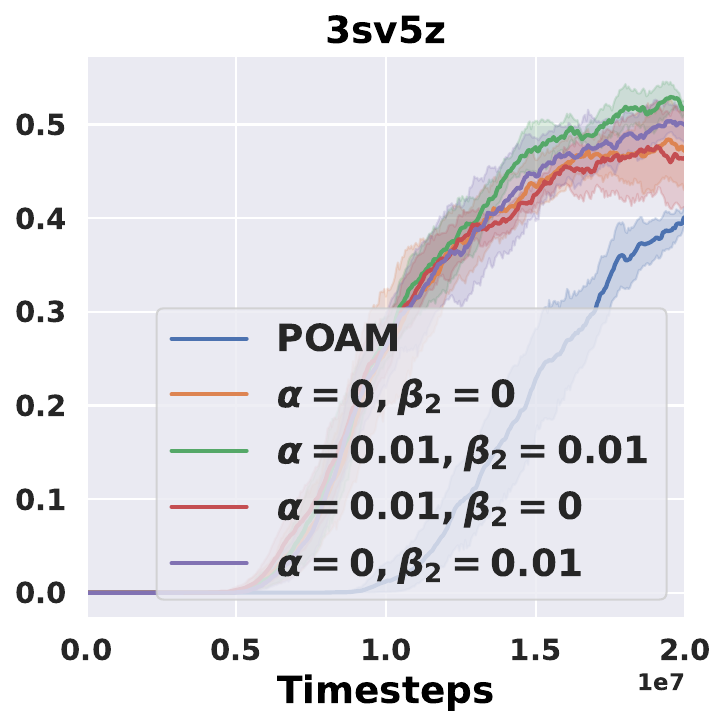}
            \includegraphics[width=0.19\linewidth]{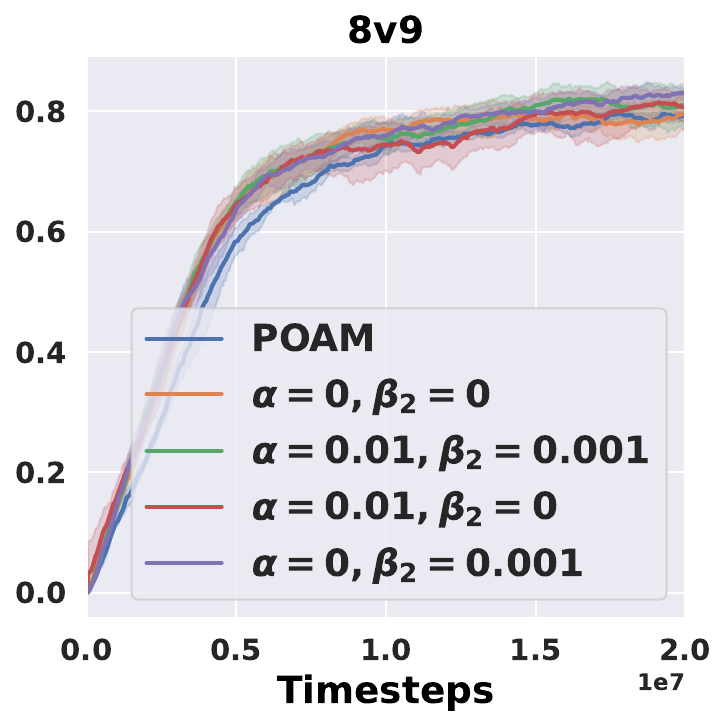}
            \includegraphics[width=0.19\linewidth]{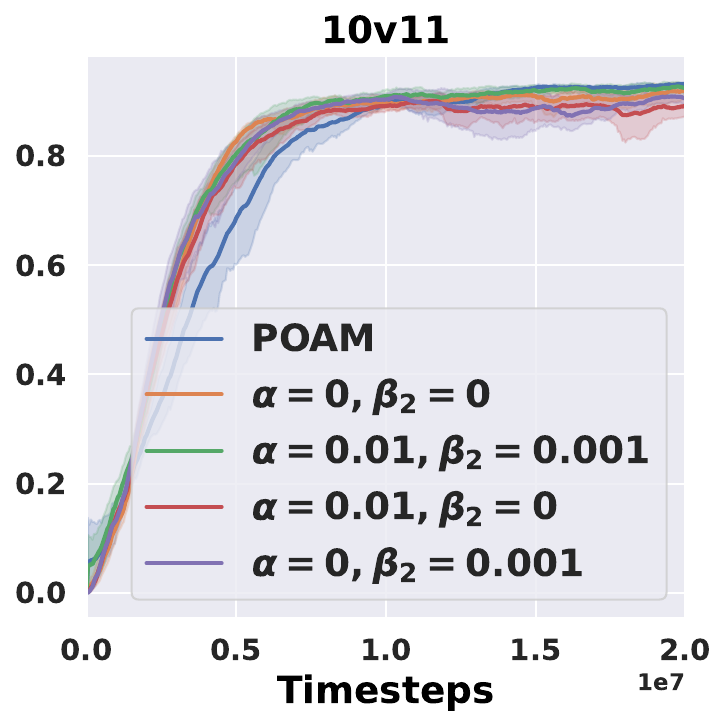}
            \caption{Ablation studies of Shapley Machine on all experimental domains. The Shapley Machine ablation variants are labelled as the $\alpha$ and $\beta_2$ settings.}
        \label{fig:ablation}
        \end{figure}

\section{Broader Impacts}
\label{sec:broader-impacts}
    The paper primarily emphasizes the positive societal impacts of its contributions, particularly in advancing multi-agent learning algorithms through principled and interpretable design. By introducing an axiomatic framework and proposing the Shapley Machine algorithm, the work has the potential to promote the development of more trustworthy and transparent AI systems, reduce redundancy in algorithm design, and enhance understanding across the field. These outcomes could ultimately support broader societal goals involving reliable AI deployment in complex cooperative environments. While the paper, being theoretical in nature, does not explicitly explore negative societal impacts, it implicitly raises concerns about potential risks, especially when such multi-agent systems are applied in sensitive domains such as military operations.

\end{document}